\documentclass{style/vldb}
\usepackage{amssymb}
\usepackage{amsmath}
\usepackage{color}
\usepackage{times}
\usepackage{cite}
\usepackage{enumitem}
\usepackage{arydshln}
\usepackage{url}
\usepackage{algorithm}
\usepackage{algorithmic}
\usepackage{flushend}
\usepackage{subcaption}
\usepackage{todonotes}
\usepackage{balance}

\usepackage{enumitem}
\setlist{nolistsep,leftmargin=*}

\title{Query Reranking As A Service}
\numberofauthors{1}
\author
{
\alignauthor
Abolfazl Asudeh$^{\ddag}$,
Nan Zhang$^{\dag\dag}$, 
Gautam Das$^{\ddag}$
\\
\affaddr {
$^{\ddag}$University of Texas at Arlington;
$^{\dag\dag}$George Washington University
}
{\email
{
$^{\ddag}$\{ab.asudeh@mavs,~gdas@cse\}.uta.edu, $^{\dag\dag}$nzhang10@gwu.edu
}
}
}
\date{}

\newcommand{\hindent}[1][1]{\hspace{#1\algorithmicindent}}

\usepackage[font={small,bf}]{caption} %changes caption size style
\begin{document}
\maketitle

%%%%%%%%%%%%%%%%%%%%%%%%%%%%%%%%%%%%%%%%%%%%%%%%%%%%%%%%%
\begin{abstract}
The ranked retrieval model has rapidly become the de facto way for search query processing in client-server databases, especially those on the web.  Despite of the extensive efforts in the database community on designing better ranking functions/mechanisms, many such databases in practice still fail to address the diverse and sometimes contradicting preferences of users on tuple ranking, perhaps (at least partially) due to the lack of expertise and/or motivation for the database owner to design truly effective ranking functions. This paper takes a different route on addressing the issue by defining a novel {\em query reranking problem}, i.e., we aim to design a third-party service that uses nothing but the public search interface of a client-server database to enable the on-the-fly processing of queries with any user-specified ranking functions (with or without selection conditions), no matter if the ranking function is supported by the database or not. We analyze the worst-case complexity of the problem and introduce a number of ideas, e.g., on-the-fly indexing, domination detection and virtual tuple pruning, to reduce the average-case cost of the query reranking algorithm. We also present extensive experimental results on real-world datasets, in both offline and live online systems, that demonstrate the effectiveness of our proposed techniques.
\end{abstract}
%%%%%%%%%%%%%%%%%%%%%%%%%%%%%%%%%%%%%%%%%%%%%%%%%%%%%%%%%
\section{Introduction}\label{sec1}
\noindent
{\bf Problem Motivation:} The ranked retrieval model has rapidly replaced the traditional Boolean retrieval model as the de facto way for query processing in client-server (e.g., web) databases.  Unlike the Boolean retrieval model which returns all tuples matching the search query selection condition, the ranked retrieval model orders the matching tuples according to an often proprietary ranking function, and returns the top-$k$ tuples matching the selection condition (with possible page-turn support for retrieving additional tuples).

The ranked retrieval model naturally fits the usage patterns of client-server databases. For example, the short attention span of clients such as web users demands the most desirable tuples to be returned first. In addition, to achieve a short response time (e.g., for web databases), it is essential to limit the length of returned results to a small value such as $k$.  Nonetheless, the ranked retrieval model also places more responsibilities on the web database designer, as the {\em ranking function} design now becomes a critical feature that must properly capture the need of database users.

In an ideal scenario, the database users would have fairly homogeneous preferences on the returned tuples (e.g., newer over older product models, cheaper over more expensive goods), so that the database owner can provide a small number of ranking functions from which the database users can choose to fulfill their individual needs. Indeed, the database community has developed many ranking function designs and techniques for the efficient retrieval of top-$k$ query answers according to a given ranking function.

The practical situation, however, is often much more complex.  Different users often have diverse and sometimes contradicting preferences on numerous factors. Even more importantly, many database owners simply lack the expertise, resources, or even motivation (e.g., in the case of government web databases created for policy or legal compliance purposes) to properly study the requirements of their users and design the most effective ranking functions. For example, many flight-search websites, including Kyak, Google Flights, Sky Scanner, Expedia, and Priceline offer limited ranking options on a subset of the attributes, that, for example, does not help ranking based on cost per mileage. Similar limitations apply to the websites such as Yahoo!~Autos (resp.~Blue Nile), if we want to rank the results, for example, based on mileage per year (resp.~summation of depth and table percent).
As a result, there is often a significant gap, in terms of both design and diversity, between the ranking function(s) supported by the client-server database and the true preferences of the database users. The objective of this paper is to define and study the {\em query re-ranking} problem which bridges this gap for real-world client-server databases.

\vspace{2mm}
\noindent{\bf Query Re-Ranking:} Given the challenge for a real-world database owner to provide a comprehensive coverage of user-preferred ranking functions, in this paper we develop a {\em third-party query re-ranking service} which uses nothing but the public search interface of a client-server database to enable the on-the-fly processing of queries with user-specified ranking functions (with or without selection conditions), no matter if the ranking function is supported by the database or not.

This query re-ranking service can enable a wide range of interesting applications.  For example, one may build a personalized ranking application using this service, offering users with the ability to remember their preferences across multiple web databases (e.g., multiple car dealers) and apply the same personalized ranking over all of them despite the lack of such support by these web databases. As another example, one may use the re-ranking service to build a dedicated application for users with disabilities, special needs, etc., to enjoy appropriate ranking over databases that do not specifically tailor to their needs.

There are two critical requirements for a solution to the query re-ranking service: First, the output query answer must precisely follow the user-specified ranking function, i.e., there is no loss of accuracy and the query re-ranking service is transparent to the end user as far as query answers are concerned. Second, the query re-ranking service must minimize the number of queries it issues to the client-server database in order to answer a user-specified query.  This requirement is crucial for two reasons: First is to ensure a fast response time to the user query, given that queries to the client-server database must be issued on the fly. Second is to reduce the burden on the client-server database, as many real-world ones, especially web databases, enforce stringent rate limits on queries from the same IP address or API user (e.g., Google Flight Search API allows only 50 free queries per user per day).

\vspace{2mm}
\noindent{\bf Problem Novelty:} While extensive studies have focused on translating an unsupported query to multiple search queries supported by a database, there has not been research on the {\em translation of ranking} requirements of queries. Related to our problem here includes the existing studies on crawling client-server databases \cite{sheng2012optimal}, as a baseline solution for query re-ranking is to first crawl all tuples from the client-server database, and then process the user query and ranking function locally. The problem, however, is the high query cost. As proved in \cite{sheng2012optimal}, the number of queries that have to be issued to the client-server database for crawling ranges from at least linear to the database size in the best-case scenario to quadratic and higher in worse cases. As such, it is often prohibitively expensive to apply this baseline to real-world client-server databases, especially those large-scale web databases that constantly change over time.

Another seemingly simple solution is for the third-party service to retrieve more than $k$ tuples matching the user query, say $h \cdot k$ tuples by using the ``page-down'' feature provided by a client-server database (or \cite{thirumuruganathan2013breaking, thirumuruganathan2013rank} when such a feature is unavailable), and then locally re-rank the $h \cdot k$ tuples according to the user-specified ranking function and return the top-$k$ ones. There are two problems with this solution. First, since many client-server databases choose not to publish the design of their proprietary ranking functions (e.g., simply naming it ``rank by popularity'' in web databases), results returned by this approach will have unknown error unless all tuples satisfying the user query are crawled. Second, when the database ranking function differs significantly from the user-specified one, this approach may have to issue many page-downs (i.e., a large $h$) in order to retrieve the real top-$k$ answers according to the user-specified ranking function.

%\textcolor{blue}{
Finally, note that our problem stands in sharp contrast with existing studies on processing top-$k$ queries over traditional databases using pre-built indices and/or materialized views (e.g., \cite{chang2000onion, PREFER}).  The key difference here is the underlying {\em data access model}: Unlike prior work which assume complete access to data, we are facing a restricted, top-$k$, search interface provided by the database.
%}

%materialize global top-k don?t make sense because they may not even be searched by users, also materialize previous query answers isn?t effective because of the numerous selection conditions
%In our previous works on hidden databases, we went beyond the $k$ provided by the Top-$k$ interface of the database, got some aggregates on it, and now we are changing the system ranking function itself.
%
%Some times the system ranking functions are not adequate for all users and some like to see things ranked differently.
%Meanwhile, it may not possible for the database owners to determine all such ranking functions beforehand. Yet many of the websites do not provide the proper ranking option. For example ???. They may not have enough resources to change their websites or may not be intensive enough to do so.
%A third party service may provide such services on the top of the website interface to increase its usability.

\vspace{2mm}
\noindent {\bf Outline of Technical Results:} We start by considering a simple instance of the problem, where the user-desired ranking function is on a single attribute, and developing Algorithm 1D-RERANK to solve it. Note that this special, 1D, case not only helps with explaining the key technical challenges of query reranking, but also can be surprisingly useful for real-world web databases. For example, a need often arising in flight search is to maximize or minimize the layover time, so as to either add a free stopover for a sightseeing day trip or to minimize the amount of agonizing time spent at an airport.  Unfortunately, while flight search websites like Kayak offer the ability to specify a range query on layover time, it does not support ranking according to the attribute.  The 1D-RERANK algorithm handily addresses this need by enabling a ``Get-Next'' primitive - i.e., upon given a user query $q$, an attribute $A_i$, and the top-$h$ tuples satisfying $q$ according to $A_i$, it finds the ``next'', i.e., $(h + 1)$-th ranked, tuple.

In the development of 1D-RERANK, we rigidly prove that, in the worst-case scenario, retrieving even just the top-1 tuple requires crawling of the entire database. Nonetheless, we also show that the practical query cost tends to be much smaller. Specifically, we found a key factor (negatively) affecting query cost to be what we refer to as ``dense regions'' - i.e., a large number of tuples clustering together within a small interval (on the attribute under consideration). The fact that a dense region may be queried again and again (by the third-party query reranker) for the processing of different user queries prompts us to propose an {\em on-the-fly indexing} idea that detects such dense regions and proactively crawls top-ranked tuples in it to avoid the waste on processing future user queries. We demonstrate theoretically and experimentally the effectiveness of such an index on reducing the overall query cost.

To solve the general problem of query reranking for any arbitrary user-desired ranking function (rather than just 1D), a seemingly simple solution is to directly apply a classic top-$k$ query processing algorithm that leverages sorted access to each attribute, e.g., Fagin's or TA algorithm \cite{fagin2003}, by calling the ``Get-Next'' primitive provided by 1D-RERANK as a subroutine. The problem with this simple solution, however, is that it incurs a significant waste of queries when applied to client-server databases, mainly because it fails to leverage the multi-predicate (conjunctive) queries supported by the underlying database. We demonstrate in the paper that this problem is particularly significant when a large number of tuples satisfying a user query feature extreme values on one or more attributes.

To address the issue, we develop MD-RERANK (i.e., Multi-Dimensional Rerank), a query re-ranking algorithm that identifies a small number of multi-predicate queries to directly retrieve the top-$k$ tuples according to a user query. We note a key difference between the 1D and MD cases: In the 1D case, a single query is enough to cover the subspace outranking a given tuple, while the MD case requires a much larger number of queries due to the more complex shape of the subspace. We develop two main ideas, namely {\em direct domination detection} and {\em virtual tuple pruning}, to significantly reduce the query cost for MD-RERANK.  In addition, like in the 1D case, we observe the high query cost incurred by ``dense regions'', and include in MD-RERANK our on-the-fly indexing idea to reduce the amortized cost of query re-ranking.

Our contributions also include a comprehensive set of experiments on real-world web databases, both in an offline setting (for having the freedom to control the database settings) and through online {\em live} experiments over real-world web databases. Specifically, we constructed a Top-$k$ web search interface in the offline experiment, and evaluated the performance of the algorithms in different situations, by varying the parameters such as database size, system-$k$, and system ranking function. In addition we also tested our algorithms live online over two popular websites, namely Yahoo!~Autos and and Blue Nile, the largest diamond online retailer. The experiment results verify the effectiveness of our proposed techniques and their superiority over the baseline competitors.

The rest of the paper is organized as follows. We provide the preliminary notions and problem definition in \S~\ref{sec2}. Then, we consider the 1D case in \S~\ref{sec3}, proving a lower bound on the worst-case query cost for query reranking and developing the on-the-fly reranking idea that significantly reduces query cost in practice for 1D-RERANK, as demonstrated in theoretical analysis.  In \S~\ref{sec4}, we study the general query reranking problem and developing the other two ideas, direct domination detection and virtual tuple pruning, for MD-RERANK. After discussing the extensions in \S~\ref{sec5}, we present a comprehensive set of experimental results in \S~\ref{sec:experiments}.  We discuss the related work in \S~\ref{sec-related}, followed by final remarks in \S~\ref{sec-final}.

%\begin{itemize}
%\item We provide a framework for finding the Top-$k$ ($k$ may be different from the $k$ provided by the hidden database) tuples on any user-specified monotonic ranking function. In order to do so, we propose an efficient method for doing the \emph{Get-next} operation on a single attribute.
%\item We introduce a semi-indexing method that is relatively compact, and still helps improving the reranking process in orders of magnitude.
%\item We propose a sampling based approach for building the indices.
%\item Considering the framework as the baseline, we provide several methods for finding the reranking the tuple in higher dimensions on the linear ranking functions.
%\end{itemize}

\section{Preliminaries}\label{sec2}
\subsection{Database Model}

\noindent
{\bf Database:} Consider a client-server database $D$ with $n$ tuples over $m$ ordinal attributes $A_1, \ldots, A_m$.  Let the value domain of $A_i$ be $V(A_i) = \{v_{i1}, \ldots, v_{i|V(A_i)|}\}$.  The database may also have other categorical attributes $B_1, \ldots, B_{m^\prime}$. But since they are usually not part of any ranking function, they are not the focus of our attention for the purpose of this paper. We assume each tuple $t$ to have a none-NULL value on each (ordinal) attribute $A_i$, which we refer to as $t[A_i]$ ($t[A_i] \in V(A_i)$). Note that if NULL values do exist in the database, the ranking function usually substitutes it with another default value (e.g., the mean or extreme value of an attribute). In that case, we simply consider the occurrence of NULL as the substituted value.  In most part of the paper, we make the general positioning assumption \cite{yale1968geometry}, before introducing a simple post-processing step that removes this assumption in \S~\ref{sec5}.

\vspace{2mm}
\noindent
{\bf Query Interface:}
Most client-server database allow users to issue certain ``simplistic'' search queries. Often these queries are limited to conjunctive ones with predicates on one or a few attributes. Examples here include web databases, which usually allows such conjunctive queries to be specified through a form-like web search interface. Formally, we consider search queries of the form

\begin{center}
$q$: SELECT * FROM D WHERE $A_{i_1} \in (v_{i_1}, v^\prime_{i_1})$ AND $\cdots$ AND $A_{i_p} \in (v_{i_p}, v^\prime_{i_p})$ AND {\em conjunctive predicates on $B_1, \ldots, B_{m^\prime}$},
\end{center}

\noindent where $\{A_{i_1}, \ldots, A_{i_p}\} \subseteq \{A_1, \ldots, A_m\}$ is a subset of ordinal attributes, and $(v_{i_j}, v^\prime_{i_j}) \subseteq V(A_{i_j})$ is a range within the value domain of $A_{i_j}$. %We use $Sel(q)$ to represent the selection predicates of $q$.

A subtle issue here is that our definition of $q$ only includes open ranges $(x, y)$, i.e., $x < A_i < y$, while real-world client-server databases may offer close ranges $[x, y]$, i.e., $x \leq A_i \leq y$, or a combination of both (e.g., $(x, y]$). We note that these minor variations do not affect the studies in this paper, because it is easy to derive the answer to $q$ even when only close ranges are allowed by database: One simple needs to find a value arbitrarily close to the limits, say $x + \epsilon$ and $y - \epsilon$ with an arbitrarily small $\epsilon > 0$, and substitute $(x, y)$ with $[x + \epsilon, y - \epsilon]$. In the case where the value domains are discrete, substitutions can be made to the closest discrete value in the domain.

% is a user-specified query $q$ provided via the web interface. The interface allows constructing the queries in the form of open and closed ranges (i.e. $<$, $>$, $\leq$, and $\geq$) for the numerical attributes, and the categorical attributes can get bound to a value, or left by $*$.
 
%As the Result, $q$ is in the form of SELECT * FROM D WHERE $Sel(q)$, where $Sel(q)$ is a set of range predicates for numerical and assignment predicates for categorical attributes.

As discussed in \S~\ref{sec1}, once a client-server database receives query $q$ from a user, it often limits the number of returned tuples to a small value $k$.  Without causing ambiguity, we use $q$ to refer to the set of tuples {\em actually returned} by $q$, $R(q)$ to refer to the the set of tuples matching $q$ (which can be a proper superset of the returned tuples $q$ when there are more than $k$ returning tuples, and $|R(q)|$ to refer to the number of tuples matching $q$.  When $|R(q)| > k$, we say that $q$ {\em overflows} because only $k$ tuples can be returned. Otherwise, if $|R(q)| \in [1, k]$, we say that $q$ returns a {\em valid} answer. At the other extreme, we say that $q$ {\em underflows} when it returns empty, i.e., $|R(q)| = 0$.

\vspace{2mm}
\noindent
{\bf System Ranking Function:} In most parts of the paper, we make a conservative assumption that, when $|R(q)| > k$, the database selects the $k$ returned tuples from $R(q)$ according to a proprietary {\em system ranking function} unbeknown to the query reranking service. That is, we make {\em no} assumption about the system ranking function whatsoever.  In \S~\ref{sec5}, we also consider cases where the database offers more ranking options, e.g., ORDER BY according to a subset of ordinal attributes.

%\textcolor{red}{open and close range discussions need to be included here}

%\textcolor{blue}{attributes w/ or w/o ORDERBY}

%\textcolor{red}{define monotonic and linear ranking functions}

%\textcolor{red}{define $R(q)$ as the set of tuples matching $q$, also define $|R(q)|$ as the count}

%\textcolor{red}{define $Sel(q)$ as the selection conditions specified in $q$, i.e., $q$ looks like SELECT * FROM D WHERE $Sel(q)$}

%\textcolor{red}{only discuss 2-ended ranges here}

%\textcolor{red}{define $t_h[A_i]$} {\bf [comment from abol] not sure if/how it is different from $t[A_i]$}

%\textcolor{red}{define the term ``system ranking function''}

%\textcolor{red}{define $k$ as in top-$k$}

%\textcolor{red}{define $\Omega(A_i)$ to be the range/domain of $A_i$}
\break
\subsection{Problem Definition}

The objective of this paper is to enable a {\em third-party query reranking service} which enables a user-specified ranking function for a user-specified query $q$, when the query $q$ is supported by the underlying client-server database but the ranking function is {\em not}.

\vspace{2mm}
\noindent {\bf User-Specified Ranking Functions:} We allow a user of the query reranking service to specify a {\em user-specified ranking function} $\mathcal{S}(q, t)$ which takes as input the user query $q$ and one or more ordinal attributes (i.e., $A_1, \ldots, A_m$) of a tuple $t$, and outputs the ranking score for $t$ in processing $q$.  The {\em smaller} the score $\mathcal{S}(q, t)$ is, the {\em higher ranked} $t$ will be in the query answer, i.e., the more likely $t$ is included in the query answer when $R(q) > k$. Without causing ambiguity, we also represent $\mathcal{S}(q, t)$ as $\mathcal{S}(t)$ when the context (i.e., the user query being processed) is clear.

We support a wide variety of user-specified ranking functions with only one requirement: {\em monotonicity}. Given a user query $q$, a ranking function $\mathcal{S}(t)$ is monotonic if and only if there exists an order of values for each attribute domain, which we represent as $\prec$ with $v_1 \prec v_2$ indicating $v_1$ being higher-ranked than $v_2$, such that there does not exist two possible tuple values $t_1$ and $t_2$ with $\mathcal{S}(t_1) < \mathcal{S}(t_2)$ yet $t_2[A_i] \prec t_1[A_i]$ for all $i \in [1, m]$.

Intuitively, the definition states that if $t_1$ outranks $t_2$ according to $\mathcal{S}(\cdot)$, then $t_1$ has to outrank $t_2$ on at least one attribute according to the order $\prec$.  In other words, $t_1$ cannot outrank $t_2$ if it is dominated \cite{chomicki} by $t_2$. Another interesting note here is that we do {\em not} require all user-specified ranking functions to follow the same attribute-value order $\prec$. For example, one ranking function may prefer higher prices while the other prefers lower prices.  We support both ranking functions so long as each is monotonic according to its own order of attribute values.

%\emph{Linear} ranking functions are a subset of \emph{monotonic} ranking functions that are  defined in the form of linear combination of the attribute values, i.e. $$c(t) = \sum_{\forall A_i \in \mathcal{A}} c_i t[A_i]$$ where $c_i$ is constant.

\vspace{2mm}
\noindent {\bf Performance Measure:} To enable query reranking, we have to issue a number of queries to the underlying client-server database. It is important to understand that the most important efficiency factor here is {\it the total number of queries} issued to the database, not the computational time. The rational behind it is that almost many client-server databases, e.g., almost all client-server databases, enforce certain {\em query-rate limit} by allowing only a limited number of queries per day from each IP address, API account, etc.

\vspace{2mm}
\noindent {\bf Problem Definition:} In this paper, we consider the problem of query reranking in a ``Get-Next'', i.e., incremental processing, fashion.  That is, for a given user query $q$, a user-specified ranking function $\mathcal{S}$, and the top-$h$ tuples satisfying $q$ according to $\mathcal{S}$, we aim to find the No.~$(h+1)$ tuple. When $h = 0$, this means finding the top-1 for given $q$ and $\mathcal{S}$.  One can see that finding the top-$h$ tuples for $q$ and $\mathcal{S}$ can be easily solved by repeatedly calling the Get-Next function.  The reason why we define the problem in this fashion is to address the real-world scenario where a user first retrieves the top-$h$ answers and, if still unsatisfied with the returned tuples, proceeds to ask for the No.~$(h+1)$. By supporting incremental processing, we can progressively return top answers while paying only the incremental cost.
%\textcolor{red}{In this paper, we consider the problem of query reranking in a ``Get-Next'' fashion.  That is, for a given user query $q$, a user-specified ranking function $\mathcal{S}$, and the top-$h$ tuples satisfying $q$ according to $\mathcal{S}$, we aim to find the No.~$(h+1)$ tuple.} One can see that finding the top-$h$ answers to $q$ according to $\mathcal{S}$ can be easily solved by repeatedly calling the Get-Next function.  The reason why we define the problem in this fashion is to address the real-world scenario where a user first retrieves the top-$h$ answers and, if still unsatisfied with the returned tuples, proceeds to ask for the No.~$(h+1)$.

\medskip\noindent
 \framebox[\columnwidth]{\parbox{0.9\columnwidth}{ \textsc{Query reranking Problem:}
Consider a client-server database $D$ with a top-$k$ interface and an arbitrary, unknown, system ranking function. Given a user query $q$, a user-specified monotonic ranking function $\mathcal{S}$, and the top-$h$ ($h \geq 0$ can be greater than, equal to, or smaller than $k$) tuples satisfying $q$ according to $\mathcal{S}$, discover the No.~$(h+1)$ tuple for $q$ while minimizing the number of queries issued to the client-server database $D$.
}}
%\textcolor{red}{define efficiency to be query cost}

%\subsection{Threshold Algorithm (TA)}
%
%Fagin et al. \cite{fagin2003} proposed \emph{TA}, the well-known algorithm for finding the Top-$k$ over any monotonic ranking function, in the \emph{sorted / random access} environment.
%They suppose the existence of a sorted list of tuples for each attribute. The algorithm starts scanning the lists in parallel, computing the scores of newly discovered tuples (based on the given ranking function), and updating the threshold, which is the minimum (best) possible score for the non-discovered tuples. It stops as soon as it discovers $k$ tuples having the score less than the threshold.
%TA targets optimizing the query cost, which in their case is equal to the number of sequential accesses required.
%
%%TA algorithm can be applied here, for reranking the $R(q)$, if we can simulate the \emph{Get-next} operation on the hidden databases, which enables the sequential access by returning the next top tuple over a specified attribute.
%In client-server databases, the interface return the all attribute-values of the tuple, thus there is no need for random access. As the result, TA can be applied here if we can enable the sequential access. Thus, in next section we concentrate on designing the efficient algorithms for reranking tuples on a single attribute.
%
%%\textcolor{red}{reiterate query cost which, in the case of TA, happens to be equal to the number of sequential accesses required by TA}
\pagebreak
\section{1D-RERANK} \label{sec3}

%Given the TA algorithm, one can enable query reranking for any monotonic ranking function so long as there is an efficient way to support ``GetNext'' on every attribute\footnote{every attribute involved in the ranking function, of course}, 

We start by considering the simple 1D version of the query reranking problem which, as discussed in the introduction, can also be surprisingly useful in practice. Specifically, for a given attribute $A_i$, a user query $q$, and the $h$ tuples having the minimum values of $A_i$ among $R(q)$ (i.e., tuples satisfying $q$), our goal here is to find tuple $t(q, A_i, h+1)$, which satisfies $q$ and has the $(h + 1)$-th smallest value on $A_i$ among $R(q)$, while minimizing the number of queries issued to the underlying database.

%As discussed in the introduction, it is the development of this algorithm which motivates us to build a global {\em density-based} index on $A_i$ that can significantly improve the efficiency of extracting $t(q, A_i, h+1)$ for all $q$ and $h$. We introduce the density-based index next in the section.

\subsection{Baseline Solution and Its Problem}

\subsubsection{1D-BASELINE}
\noindent{\bf Baseline Design:} Since our focus here is to discover $t(q, A_i, h+1)$ given $q$, $A_i$ and $h$, without causing ambiguity, we use $t_{h+1}$ as a short-hand representation of $t(q, A_i, h+1)$.  A baseline solution for finding $t_{h+1}$ is to start with issuing to the underlying database query $q_1$: SELECT * FROM D WHERE $A_i > t_h[A_i]$ AND $Sel(q)$, where $Sel(q)$ represents all selection conditions specified in $q$.  If $h = 0$, this query simply becomes SELECT * FROM D WHERE $Sel(q)$.

Note that the answer to $q_1$ must return non-empty, because otherwise it means there are only $h$ tuples matching $q$.  Let $a_1$ be the one having minimum $A_i$ among all returned tuples. Given $a_1$, the next query we issue is $q_2$: WHERE $A_i \in (t_h[A_i], a_1[A_i])$ AND $Sel(q)$.  In other words, we narrow the search region on $A_i$ to ``push the envelop'' and discover any tuple with even ``better'' $A_i$ than what we have seen so far.

If $q_2$ returns empty, then $t_{h+1} = a_1$.  Otherwise, we can construct and issue $q_3$, $q_4$, $\ldots$, in a similar fashion. More generally, given $a_j$ being the tuple with minimum $A_i$ returned by $q_j$, the next query we issue is $q_{j+1}$: WHERE $A_i \in (t_h[A_i], a_j[A_i])$ AND $Sel(q)$. We stop when $q_{j+1}$ returns empty, at which time we conclude $t_{h+1} = a_j$. Algorithm~\ref{alg:1D-baseline}, 1D-BASELINE, depicts the pseudo-code of this baseline solution.

\vspace{2mm}
\noindent{\bf Leveraging History:} An implementation issue worth noting for 1D-BASELINE is how to leverage the historic query answers we have already received from the underlying client-server database. This applies not only during the processing of a user query, but also across the processing of different user queries. 

During the process of user query $q$, for example, we do not have to start with the range of $A_i \in (t_h[A_i], \infty)$ as stated in the basic algorithm design. Instead, if we have already ``seen'' tuples in $R(q)$ that have $A_i > t_h[A_i]$ in the historic query answers, then we can first identify such a tuple with the minimum $A_i$, denoted by $t^\prime$, and then start the searching process with $A_i \in (t_h[A_i], t^\prime)$, a much smaller region that can yield significant query savings, as shown in the query cost analysis below.

More generally, this exact idea applies across the processing of different user queries.  What we can do is to inspect every tuple we have observed in historic query answers, identify those that match the user query being processed, and order these matching tuples according to the attribute $A_i$ under consideration. By doing so, the more queries we have processed, the more likely we can prune the search space for $t_{h+1}$ based on historic query answers, and thereby reduce the query cost for re-ranking. %\textcolor{red}{The pseduocode of 1D-BASELINE reflects the leveraging of historic query answers}.

%The correctness of the algorithm is because adding the predicate $A<t[A]$ guarantees that 1) no tuple in the database is discovered more than once, and 2) every discovered tuple at step $i$ has a lower value on $A$ than all the tuples discovered at steps $1$ to $i-1$. Since the algorithm only applies the upper-bound predicates, if the query $i+1$ returns null, it guarantees that there is no tuple in the database having a lower value on $A$ than the $i^{th}$ discovered tuple. Directly from the facts that each tuple is discovered at most once and there is only one query that returns NULL, the query cost of the algorithm is $O(n)$, where $n$ is the total number of the tuples in the database.
%
%In order to extend the algorithm to the \emph{Get-next} operation, after discovering the Top${_i}^{th}$ tuple, we only need to explore the area between Top$_i$ and current candidate for Top$_{i+1}$. Therefore we initiate the algorithm with SELECT * WHERE $P$ and Top$_i[A]<A<$ CurrentT$_{i+1}[A]$. An interesting property of this algorithm is that \emph{no matter how many Get-next operations we issue, the summation of the cost of all the operations is still $O(n)$}. That is because the algorithm never discovers a tuple twice, and it spends one ``NULL-returning'' query for each \emph{Get-next operation} (note that at most $n+1$ \emph{Get-next} operations are possible).

\begin{algorithm}[!htb]
\caption{{\bf 1D-BASELINE}}
\begin{algorithmic}[1]
\label{alg:1D-baseline}
\STATE $t_{h+1}$ = argmin$_{t[A_i]}\{t\in$ history $|$ $t[A_i] > t_h[A_i]\}$ 
\STATE $T$ = Top-$k$(WHERE $t_{h+1}[A_i]>A_i > t_h[A_i]$ AND $Sel(q)$)
\STATE {\bf while} $T$ is overflow
    \STATE \hindent $t_{h+1}$ = argmin$_{t[A_i]}\{t\in T\}$ 
    \STATE \hindent $T$ = Top-$k$(WHERE $t_{h+1}[A_i]>A_i > t_h[A_i]$ AND $Sel(q)$)
\STATE {\bf return} $t_{h+1}$
\end{algorithmic}
\end{algorithm}

\subsubsection{Negative Result: Lower Bound on Worst-Case Query Cost}

While simple, 1D-BASELINE has a major problem on query cost, as it depends on the correlation between $A_i$ and the system ranking function which we know nothing about and has no control over. For example, if the system ranking function is exactly according to $A_i$, then the query cost of finding $t_{h+1}$ is 2: $q_1$ returns $t_{h+1}$ and $q_2$ returns empty to confirm that $t_{h+1}$ is indeed the ``next'' tuple.  On the other hand, if the system ranking function is the exact opposite to $A_i$ (i.e., returning tuples with maximal $A_i$ first), then the query cost for the baseline solution is exactly $|R(q)| + 1$ in the worst-case scenario (when $k = 1$), because every tuple satisfying $q$ will be returned before $t_{h+1}$ is revealed at the end.  Granted, this cost can be ``amortized'' thanks to the leveraging-history idea discussed above, because the $|R(q)| + 1$ queries indeed reveal not just the top-$(h+1)$ but the complete ranking of all tuples matching $q$. Nonetheless, the query cost is still prohibitively high when $q$ matches a large number of tuples.

While it might be tempting to try to ``adapt to'' such ill-conditioned system ranking functions, the following theorem actually shows that the problem is not fixable in the worst-case sense. Specifically, there is a lower bound of $n/k$ on the query cost required for query reranking given the worst-case data distribution and worst-case system ranking function.

\newtheorem{theorem}{Theorem}
\begin{theorem} \label{thm:ner}
$\forall n > 1$, there exists a database of $n$ tuples such that finding the top-ranked tuple on an attribute through a top-$k$ search interface requires at least $n/k$ queries that retrieve all the $n$ tuples.
\end{theorem}
\begin{proof}
Without loss of generality, consider a database with only one attribute $A$ and an unknown ranking function.  Let $(v_0, v_\infty)$ be the domain of $A$.  Note that this means (1) the query re-ranking algorithm can only issue queries of the form SELECT * FROM D WHERE $A \in (v_1, v_2)$, where $v_0 \leq v_1 < v_2 \leq v_\infty$, (2) the returned tuples will be ranked in an arbitrary order, and (3) the objective of the query re-ranking algorithm is to find the tuple with the smallest $A$.

For any given query re-ranking algorithm $\mathcal{R}$, consider the following query processing mechanism $\mathcal{Q}$ for the database: During the processing of all queries, we maintain a min-query-threshold $v_q$ with initial value $v_\infty$. If a query $q$ issued by $\mathcal{R}$ has lower bound {\em not} equal to $v_0$, i.e., $q$: WHERE $A \in (v_1, v_2)$ with $v_1 > v_0$, $\mathcal{Q}$ returns whatever tuples already returned in historic query answers that fall into range $(v_1, v_2)$. It also sets $v_q = \min(v_q, v_1)$.

Otherwise, if $q$ is of the form WHERE $A \in (v_0, v_2)$ with $v_2 > v_0$, then $\mathcal{Q}$ returns an overflowing answer with $k$ tuples.  These $k$ tuples include those in the historic query answers that fall into $(v_0, v_2)$. If more than $k$ such tuples exist in the history, we choose an arbitrary size-$k$ subset. If fewer than $k$ such tuples exist, we fill up the remaining slots with arbitrary values in range $((v_0 + v_q)/2, v_q)$\footnote{Note that any factor here (besides 2) works too.  So in general the range can be $((v_0 + v_q) \cdot \alpha, v_q)$ so long as $\alpha > 0$.}. We also set $v_q$ to be $(v_0 + v_q)/2$.

There are two critical observations here. First is that for any query sequence $q_1, \ldots, q_h$ with $h \leq n/k$, we can always construct a database $D$ of at most $n$ tuples, such that the query answers generated by $\mathcal{Q}$ are consistent with what $D$ produces. Specifically, $D$ would simply be the union of all tuples returned. Note that our maintenance of $v_q$ ensures the consistency.

The second critical observation is that no query re-ranking algorithm $\mathcal{R}$ can find the tuple with the smallest $A$ without issuing at least $n/k$ queries. The reason is simple: since $n/k - 1$ queries cannot reveal all $n$ tuples, we can add a tuple $t$ with $A = (v_0 + v_q)/2$ to the database, where $v_q$ is its value after processing all $n/k - 1$ queries. One can see that the answers to all $n/k-1$ queries can remain the same. As such, for any $n > 1$, there exists a database containing $n$ tuples such that finding the top-ranked one for an attribute requires at least $n/k$ queries, which according to \cite{sheng2012optimal} is sufficient for crawling the entire database in a 1D space.
\end{proof}

\subsection{1D-RERANK}

Given the above result, we have to shift our attention to reducing the cost of finding $t_{h+1}$ in an average-case scenario, e.g., when the tuples are more or less uniformly distributed on $A_i$ (instead of forming a highly skewed distribution as constructed in the proof of Theorem~\ref{thm:ner}). To this end, we start this subsection by considering a binary-search algorithm.  After pointing out the deficiency of this algorithm when facing certain system ranking functions, we introduce our idea of on-the-fly indexing for the design of 1D-RERANK, our final algorithm for query reranking with a single-attribute user-specified ranking function.

\subsubsection{1D-BINARY and its Problem}
The binary search algorithm departs from 1D-BASELINE on the construction of $q_2$: Given $a_1$, instead of issuing $q_2$: WHERE $A_i \in (t_h[A_i], a_1[A_i])$ AND $Sel(q)$, we issue here
\begin{align*}
q^\prime_2: \mbox{ WHERE } A_i \in (t_h[A_i], (a_1[A_i] + t_h[A_i])/2) \mbox{ AND } Sel(q).
\end{align*}
This query has two possible outcomes: If it returns non-empty, we consider the returned tuple with minimum $A_i$, say $a_2$, and construct $q^\prime_3$ according to $a_2$.  The other possible outcome is for $q^\prime_2$ to return empty. In this case, we issue $q^{\prime\prime}_2$: WHERE $A_i \in [(a_1[A_i] + t_h[A_i])/2, a_1[A_i])$ AND $Sel(q)$, which has to return non-empty as otherwise $t_{h+1} = a_1$.  In either case, the {\em search space} (i.e., the range in which $t_{h+1}$ must reside) is reduced by at least half. Algorithm~\ref{alg:1D-binary}, 1D-BINARY, depicts the pseudocode.

\vspace{2mm}
\begin{algorithm}[!htb]
\caption{{\bf 1D-BINARY}}
\begin{algorithmic}[1]
\label{alg:1D-binary}
\STATE $t_{h+1}$ = argmin$_{t[A_i]}\{t\in$History $|$ $t[A_i] > t_h[A_i]\}$ 
\STATE {\bf do}
    \STATE \hindent $q\prime = $ WHERE $A_i \in (t_h[A_i], (t_{h+1}[A_i] + t_h[A_i])/2)$ AND $Sel(q)$
    \STATE \hindent $T$ = Top-$k$($q\prime$)
    \STATE \hindent {\bf if} $T$ is underflow
	    \STATE \hindent \hindent $q\prime$ = WHERE $A_i \in [(t_{h+1}[A_i] + t_h[A_i])/2,t_{h+1}[A_i])$ AND $Sel(q)$     
	    \STATE \hindent \hindent $T$ = Top-$k$($q\prime$)  
    \STATE \hindent {\bf if} $T$ is not underflow
	    \STATE \hindent \hindent $t_{h+1}$ = argmin$_{t[A_i]}\{t\in T\}$  
\STATE {\bf while} $T$ is overflow
\STATE {\bf return} $t_{h+1}$
\end{algorithmic}
\end{algorithm}

\vspace{2mm}
\noindent{\bf Query Cost Analysis:} While the design of 1D-BINARY is simple, the query-cost analysis of it yields an interesting observation which motivates the indexing-based design of our final 1D-RERANK algorithm.  Let
\begin{align}
\epsilon_k = t_{h+k+1}[A_i] - t_{h+1}[A_i].
\end{align}
An important observation here is that the execution of 1D-BINARY must conclude when the search space is reduced to width smaller than $\epsilon_k$, because no such range can cover $t_{h+1}[A_i]$ while matching more than $k$ tuples.  Thus, the worst-case query cost of 1D-BINARY is
\begin{align}
O(\min(\log_2(|V(q, A_i)|/\epsilon_k), |R(q)|/k)), \label{equ:bsq}
\end{align}
where $|V(q, A_i)|$ is the range of $A_i$ among tuples satisfying $q$ - i.e., $\max_{t \in R(q)} t[A_i] - \min_{t \in R(q)} t[A_i]$. Note that the second input to the $\min$ function in (\ref{equ:bsq}) is because every pair of queries issued by 1D-BINARY, i.e., $q^\prime_j$ and $q^{\prime\prime}_j$, must return at least $k$ tuples never seen before that satisfies $q$.

The query-cost bound in (\ref{equ:bsq}) illustrates both the effectiveness and the potential problem of Algorithm 1D-BINARY. On one hand, one can see that 1D-BINARY performs well when the tuples matching $q$ are uniformly distributed on $A_i$, because in this case the expected value of $\epsilon_k$ becomes $k \cdot |V(q, A_i)|/|R(q)|$, leading to a query cost of $O(\log_2 (|R(q)|/k))$.

On the other hand, 1D-BINARY still incurs a high query cost (as bad as $\Omega(|R(q)|/k)$, just as indicated by Theorem~\ref{thm:ner}) when two conditions are satisfied: (1) the system ranking function is ill-conditioned, i.e., negatively correlated with $A_i$, and (2) Within $R(q)$ there are {\em densely clustered} tuples with extremely close values on $A_i$, leading to a small $\epsilon_k$. Unfortunately, once the two conditions are met, the high query cost 1D-BINARY is likely to be incurred again and again for different user queries $q$, leading to an expensive reranking service. It is this observation which motivates our index-based reranking idea discussed next.

\subsubsection{Algorithm 1D-RERANK: On-The-Fly Indexing}

\noindent{\bf Oracle-based Design:} According to the above observation, densely clustered tuples cause a high query cost of 1D-BINARY. To address the issue, we start by considering an ideal scenario where there exists an oracle which identifies these ``dense regions'' and reveals the tuple with minimum $A_i$ in these regions without costing us any query. Of course, no such oracle exists in practice. Nevertheless, what we shall do here is to analyze the query cost of 1D-BINARY given such an oracle, and then show how this oracle can be ``simulated'' with a low-cost on-the-fly indexing technique.

%In light of the two observations, our idea of index-based reranking is to proactively {\em record as an index} all densely located tuples (on $A_i$) once we encounter it, so that we do not need to incur a high query cost every time a query $q$ triggers visits to the same dense region.

%\textcolor{red}{Abol: please make sure $V(A_i)$ is defined in prelim.}

Specifically, for any given region $[x, y] \in V(A_i)$, we call it a {\em dense region} if and only if it covers at least $s$ tuples {\em and} $y - x < |V(A_i)| \cdot (s/n) / c$, where $c$ and $s$ are parameters. In other words, the density of tuples in $[x, y]$ is more than $c$ times higher than the uniform distribution (which yields an expected value of $E(y-x) = |V(A_i)| \cdot (s/n)$). The setting of $c$ and $s$ is a subtle issue which we specifically address at the end of this subsection.  Given the definition of dense region, the oracle functions as follows: Upon given a user query $q$, an attribute $A_i$, and a range $[x, y] \subseteq V(A_i)$ as input, the oracle either returns empty if $[x, y]$ is not dense, or a tuple $t$ which (1) satisfies $q$, (2) has $A_i \in [x, y]$, and (3) features the smallest $A_i$ among all tuples satisfying (1) and (2).

With the existence of this oracle, we introduce a small yet critical revision to 1D-BINARY, by {\em terminating} binary search whenever the width of the search space becomes narrower than the threshold for dense region, i.e., $\epsilon_k < |V(A_i)| \cdot (s/n) / c$. Then, we call the oracle with the remaining search space as input.  Note that doing so may lead to two possible returns from the oracle:

One is when the region is indeed dense. In this case, the oracle will directly return us $t_{h+1}$ with zero cost. The other possible outcome is an empty return, indicating that the region is not really dense, instead containing more than $k$ (otherwise 1D-BINARY would have already terminated) but fewer than $s$ tuples. Note that this is not a bad outcome either, because it means that by following the baseline technique (1D-BASELINE) on the remaining search space, we can always find $t_{h+1}$ within $O(s/k)$ queries.

%\textcolor{red}{Abol: please structure the pseudocode in two parts: (1) 1D-RERANK, (2) Oracle. At the beginning of the oracle part, say sth like if the oracle exists, return the answer, otherwise continue (with the on-the-fly indexing pseudocode, of course).}

Algorithm~\ref{alg:1D-rerank} depicts the pseudocode of 1D-RERANK, the revised algorithm. The following theorem shows its query cost, which follows directly from the above discussions.

\vspace{2mm}
\begin{algorithm}[!htb]
\caption{{\bf 1D-RERANK}}
\begin{algorithmic}[1]
\label{alg:1D-rerank}
\STATE $t_{h+1}$ = argmin$_{t[A_i]}\{t\in$History $|$ $t[A_i] > t_h[A_i]\}$ 
\STATE {\bf while} $(t_{h+1}[A_i] - t_h[A_i]) < |V(A_i)| \cdot (s/n) / c$
    \STATE \hindent $q\prime$ = WHERE $A_i \in (t_h[A_i], (t_{h+1}[A_i] + t_h[A_i])/2)$ AND $Sel(q)$
    \STATE \hindent $T$ = Top-$k$($q\prime$)
    \STATE \hindent {\bf if} $T$ is underflow
	    \STATE \hindent \hindent  $q\prime$ = WHERE $A_i \in [(t_{h+1}[A_i] + t_h[A_i])/2,t_{h+1}[A_i])$ AND $Sel(q)$
	    \STATE \hindent \hindent $T$ = Top-$k$($q\prime$)  
    \STATE \hindent {\bf if} $T$ is not underflow   
	    \STATE \hindent \hindent $t_{h+1}$ = argmin$_{t[A_i]}\{t\in T\}$  
	\STATE \hindent {\bf if} $T$ is valid {\bf break}
\STATE {\bf if} $T$ is valid
	\STATE \hindent look up $t_{h+1}$ at ORACLE($A_i$,$(t_h[A_i],t_{h+1}[A_i])$,$q$)
\STATE {\bf return} $t_{h+1}$
\end{algorithmic}
\end{algorithm}

\begin{theorem}~\label{thm:qc}
The query cost of 1D-RERANK, with the presence of the oracle, is $O(\log(c \cdot n/s) + s/k)$.
\end{theorem}

\begin{proof}
The query cost of of 1D-RERANK, with the presence of the oracle, is the summation of the following costs:
\begin{itemize}
\item $c_1$: the query cost of following 1D-BINARY, until the search space becomes narrower than the dense region threshold,
\item $c_2$: the query cost of discovering $t_{h+1}$ in the remaining region, using the oracle.
\end{itemize}

Following 1D-BINARY takes $O(\log_2(|V(q, A_i)|/\epsilon_k))$ queries. Because $\epsilon_k < |V(A_i)| \cdot (s/n) / c$, $c_1$ is in the order of $O(\log(c \cdot n/s))$.
As discussed previously, if the oracle does not include the remaining region, the region is not dense and contains fewer than $s$ tuples. Then, following 1D-BASELNE, at most $s/k$ queries are requires to discover $t_{h+1}$, i.e. $c_2$ is $O(s/k)$.
Consequently, the query cost of 1D-RERANK, with the presence of the oracle, is $O(\log(c \cdot n/s) + s/k)$.
\end{proof}

%\textit{Due to lack of space, the proof for Theorem~\ref{thm:qc} is omitted.}

Note that the query cost indicated by the theorem is very small. For example, when $c = n$ and $s = k \cdot \log n$, the query cost is $O(\log n)$, substantially smaller than that of 1D-BINARY. Of course, the oracle does not exist in any real system. Thus, our goal next is to simulate this oracle with an efficient on-the-fly indexing technique.

%we add the crawled tuples into the index. We also attribute the query cost incurred by the crawling process to the cost of {\em indexing}, reflecting the fact that it is likely an amortized cost that can make future rerankings more efficient. 

%to directly obtain  instead of continuing with the binary search process.

%algorithm can be described as follows. We follow the design of 1D-BINARY exactly until the search space becomes (1) covered by an already crawled region in the index; or (2) narrower than the above-mentioned bound $|V(A_i)| \cdot (s/n) / c$ (we shall discuss the selection of input parameters $s$ and $c$ momentarily). If the search space has been crawled already, we can directly conclude 1D-RERANK based on the crawled tuples.

%Otherwise, if the remaining search space becomes narrower than the threshold, we start calling on 1D-BASELINE with $q$ being SELECT * FROM D to crawl all tuples in the region.  

\vspace{2mm}
\noindent{\bf On-The-Fly Indexing:} Our idea for simulating the oracle is simple: once 1D-RERANK decides to call the oracle with a range $(x, y)$, we invoke the 1D-BASELINE algorithm on SELECT * FROM D WHERE $A_i \in (x, y)$ to find the tuple $t$ with smallest $A_i$ in the range. If $t$ satisfies the user query $q$ being processed, then we can stop and output $t$. Otherwise, we call 1D-BASELINE on WHERE $A_i \in (t[A_i], y)$ to find the No.~2 tuple, and repeat this process until finding one that satisfies $q$.  All tuples discovered during the process are then added into the ``dense index'' that is maintained throughout the processing of all user queries.

Algorithm~\ref{alg:oracle} depicts the on-the-fly index building process. Note that the index we maintain is essentially a set of 3-tuples
\begin{align}
\langle A_i, (x, y), D(A_i, x, y))\rangle,
\end{align}
where $A_i$ is an attribute, $(x, y)$ is a range in $V(A_i)$ (non-overlapping with other indexed ranges of $A_i$), and $D(A_i, x, y)$ contains all (top-ranked) tuples we have discovered that have $A_i \in (x, y)$.

\vspace{2mm}
\begin{algorithm}[!htb]
\caption{{\bf ORACLE}}
\begin{algorithmic}[1]
\label{alg:oracle}
\STATE {\bf if} ORACLE($A_i$,$x,y$) exists
	\STATE \hindent {\bf return} argmin$_{t[A_i]}\{t\in D(A_i, x, y))|$ $t$ matches $Sel(q)\}$
\STATE $t$=1D-BASELINE(WHERE $A_i \in (x, y)$)
\STATE add $t$ to $D(A_i, x, y)$
\STATE {\bf while} $t$ does not satisfy $Sel(q)$
	\STATE \hindent $t$=1D-BASELINE(WHERE $A_i \in (t[A_i], y)$)
	\STATE \hindent add $t$ to $D(A_i, x, y)$
\STATE {\bf return} $t$
\end{algorithmic}
\end{algorithm}

Note that this simulation does differ a bit from the ideal oracle. Specifically, it does not really determine if the region is dense or not.  Even if the region is not dense, this simulated oracle still outputs the correct tuple. What we would like to note, however, is that this difference has no implication whatsoever on the query cost of 1D-RERANK. Specifically, what happens here is simply that the on-the-fly indexing process pre-issues the queries 1D-RERANK is supposed to issue when the oracle returns empty. The overall query cost remains exactly the same.

Another noteworthy design in on-the-fly indexing is the call of 1D-BASELINE on SELECT * FROM D WHERE $A_i \in (x, y)$, a query that does not ``inherit'' the selection conditions in the user query $q$ being processed. This might appear like a waste as 1D-BASELINE could issue fewer queries with a narrower input query. Nonetheless, we note that rationale here is that a dense region might be covered by multiple user queries repeatedly. By keeping the index construction generic to all user queries, we reduce the amortized cost of indexing as the dense index can make future reranking processes more efficient.

%the cost of {\em indexing}, reflecting the fact that it is likely an amortized cost that can make future rerankings more efficient. 

%On the other hand, the region might turn out to be actually sparse. In this case, while we can still add the crawled tuples to the index (for further usage), we attribute the query cost incurred by crawling to that of reranking $q$, not the ``amortized'' indexing cost. This seemingly subtle ``accounting gimmick'' can help us with analyzing the true cost of 1D-RERANK and determining the parameter settings for $c$ and $s$, as shown in the following theorem.

\vspace{2mm}
\noindent{\bf Parameter Settings:} To properly set the two parameters for dense index, $c$ and $s$, we need to consider not only the query cost derived in Theorem~\ref{thm:qc}, but also the cost for building the index, which is considered in the following theorem:

\begin{theorem}\label{th:3}
The total query cost incurred by on-the-fly indexing (for processing all user queries) is at most
\begin{align}\label{eq:onTheFlyCost}
\sum^{n-s-1}_{h=1} c(h)
\end{align}
where $c(h) = 1$ if there exists $j \in [h - s, h]$, such that
\begin{align} \label{eq:densewindow}
t(*, A_i, j+s+1)[A_i] - t(*, A_i, j)[A_i] < \frac{s \cdot |V(A_i)|}{c \cdot n},
\end{align}
and 0 otherwise.  Here $t(*, A_i, j)$ refers to the $j$-th ranked tuple according to $A_i$ in the entire database.
\end{theorem}

\begin{proof}
The discovery of every tuple in the dense region takes at most the amortized cost of one query. That is because 1D-BASELINE assures the discovery on $k$ unseen tuples by every non-underflowing query, i.e. every tuple in the dense region is discovered by one and only one query. Thus the query cost is at most equal to the number of tuples in the dense regions.
Each tuple $t$ is in the dense region with regard to the dimension $A_i$, if, sorting the tuples on $A_i$, we can construct a window containing $t$, with size less than the dense region threshold, that has at least $s$ tuples. Suppose $t$ is ranked $h$-th based on $A_i$. Equation~\ref{eq:densewindow} checks the existence of such a window around it. The total cost thus, is at most the number of the tuples for which this equation is true. This is reflected in Equation~\ref{eq:onTheFlyCost}.
\end{proof}

%\textit{Due to lack of space, the proof for Theorem~\ref{th:3} is omitted.}

One can see from the above theorem and Theorem~\ref{thm:qc} how $c$ and $s$ impacts the query cost: the larger $c$ is, the fewer dense regions there will be, leading to a lower indexing cost. On the other hand, the per-query reranking cost increases at the log scale with $c$. Similarly, the larger $s$ is, the fewer dense regions there will be (because a larger $s$ reduces the variance of tuple density), while the per-query reranking cost increases linearly with $s$. Given the different rate of increase for the per-query reranking cost with $c$ and $s$, we should set $c$ to be a larger value to leverage its log-scale effect, while keep $s$ small to maintain an efficient reranking process.

Specifically, we set $c = n$ and $s = k \cdot \log n$. One can see that the per-query reranking cost of 1D-RERANK in this case is $O(\log n)$. While the indexing cost depends on the specific data distribution (after all, we are bounded by Theorem~\ref{thm:ner} in terms of worst-case performance), the large value of $c = n$ makes it extremely unlikely for the indexing cost to be high. %\textcolor{red}{
In particular, note that even if the density surrounding each tuple follows a heavy-tailed scale-free distribution, the setting of $c = n$ still makes the number of dense regions, therefore the query cost for indexing, a constant.%}
We shall verify this intuition and perform a comprehensive test of different parameter settings in the experimental evaluations.

\section{MD-RERANK} \label{sec4}

In this section, we consider the generic query reranking problem, i.e., over any monotonic user-specified ranking function.  We start by pointing out the problem of a seemingly simple solution: implementing a classic top-$k$ query processing algorithm such as TA \cite{fagin2003} by calling 1D-RERANK as a subroutine.  The problem illustrates the necessity of properly leveraging the conjunctive queries supported by the search interface of the underlying database.  To do so, we start with the design of MD-BASELINE, a baseline technique similar to 1D-BASELINE. Despite of the similarity, we shall point out a key difference between two cases: MD-BASELINE requires many more queries because of the more complex shape of what we refer to as a tuple's ``rank-contour'' - i.e., the subspace (e.g., a line in 2D space) containing all possible tuples that have the same user-defined ranking score as a given tuple $t$. To reduce this high query cost, we propose Algorithm MD-BINARY which features two main ideas, {\em direct domination detection} and {\em virtual tuple pruning}. Finally, we integrate the dense-region indexing idea with MD-BINARY to produce our final MD-RERANK algorithm.

\subsection{Problem with TA over 1D-RERANK}

\begin{figure}
\center
\includegraphics[width=0.35\textwidth]{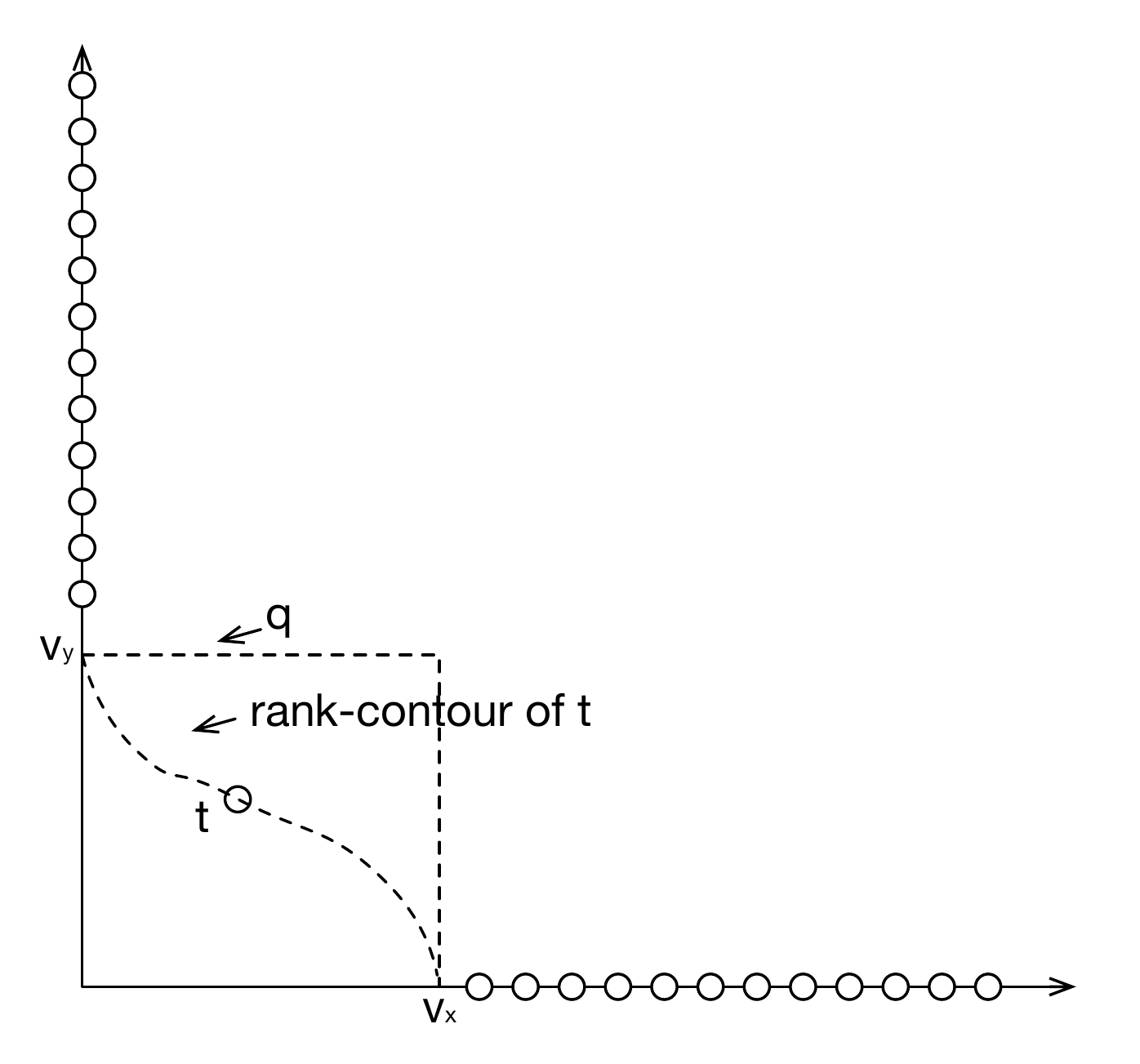}
\caption{Illustration of problem with TA over 1D-RERANK}
\label{fig:1DN}
\end{figure}

%\begin{figure*}[ht]
%    \begin{minipage}[t]{0.33\linewidth}
%        \centering
%        %\includegraphics[width = 40mm]{plots/RQ-N.pdf}
%        \includegraphics[scale=0.35]{figures/md-ta.pdf}
%		\vspace{-3mm}\caption{Illustration of problem with TA over 1D-RERANK}
%        \label{fig:1DN}
%    \end{minipage}
%    \hspace{-2mm}
%    \begin{minipage}[t]{0.33\linewidth}
%        \centering
%        \includegraphics[scale=0.35]{figures/md-exp.pdf}
%		\vspace{-3mm}\caption{Example of MD-BASELINE}
%		\label{fig:2DE}
%    \end{minipage}
%    \hspace{-2mm}
%    \begin{minipage}[t]{0.33\linewidth}
%        \centering
%        \includegraphics[scale=0.35]{figures/md-binary.pdf}
%		\vspace{-3mm}\caption{Illustration of problem with MD-Baseline}
%		\label{fig:2D1}
%    \end{minipage}
%    %\hspace{-2mm}
%\end{figure*}

A seemingly simple solution to solve the generic query reranking problem is to directly apply a classic top-$k$ query processing algorithm, e.g., the threshold (TA) algorithm \cite{fagin2003}, over the Get-Next primitive offered by 1D-RERANK.  While we refer readers to \cite{fagin2003} for the detailed design of TA, it is easy to see that 1D-RERANK offers all the data structure required by TA, i.e., a sorted access to each attribute.  Note that the random access requirement does not apply here because, as discussed in the preliminary section, the search interface returns all attribute values of a tuple without the need for accessing each attribute separately.  Since TA supports all monotonic ranking functions, this simple combination solves the generic query reranking problem defined in \S~\ref{sec2}.

While simple, this solution suffers from a major efficiency problem, mainly because it does not leverage the full power provided by client-server databases. Note that, by exclusively calling 1D-RERANK as a subroutine, this solution focuses on just one attribute at a time and does not issue any multi-predicate (conjunctive) queries supported by the underlying database (unless such predicates are copied from the user query). The example in Figure~\ref{fig:1DN} illustrates the problem: In the example, there is a large number of tuples with extreme values on both attributes (i.e., tuples on the $x$- and $y$-axis). Since this TA-based solution focuses on one attribute at a time, these extreme-value tuples have to be enumerated first {\em even when} the system ranking function completely aligns (e.g., equals) the user-desired ranking function. In other words, no matter what the system/user ranking function is, discovering the top-1 reranked tuple requires sifting through at least half of the database in this example.

On the other hand, one can observe from the figure the power bestowed by the ability to issue multi-predicate conjunctive queries. As an example, consider the case where the system ranking function is well-conditioned and returns $t$ as the result for SELECT * FROM D.  Given $t$, we can compute its rank-contour, i.e., the line/curve that passes through all 2D points with user-defined score equal to $\mathcal{S}(t)$, i.e., the score of $t$. The curve in the figure depicts an example. Given the rank-contour, we can issue the smallest 2D query encompassing the contour, e.g., $q$ in Figure~\ref{fig:1DN}, and immediately conclude that $t$ is the No.~1 tuple when $q$ returns $t$ and nothing else (assuming $k > 1$). This represents a significant saving from the query cost of implementing TA over 1D-RERANK.

\subsection{MD-Baseline}

\subsubsection{Discovery of Top-1}

To leverage the power of multi-predicate queries, we start with developing a baseline algorithm similar to 1D-BASELINE.  The algorithm starts with discovering the top-1 tuple $t$ according to an arbitrary attribute, say $A_1$.  Then, we compute the rank-contour of $t$ (according to the user ranking function, of course), specifically the values where $t$'s rank-contour intersects with each dimension, i.e., 
\begin{align}
\ell(A_i) = \max\{v \in V(A_i) | \mathcal{S}(t) \leq \mathcal{S}(0, \ldots, 0, v, 0, \ldots, 0)\}.
\end{align}
Figure~\ref{fig:2DE} depicts an example of $\ell(A_i)$ for the two dimensions, computed according to $t$.

We now issue $m$ queries of the form
\begin{align}
q_1: \mbox{  } &A_1 < t[A_1] \mbox{ \& } A_2 < \ell(A_2) \mbox{ \& } \cdots \mbox{ \& } A_m < \ell(A_m) \nonumber\\
q_2: \mbox{  } &A_1 \in [t[A_1], \ell(A_1)) \mbox{ \& } A_2 < t[A_2] \mbox{ \& } A_3 < \ell(A_3) \mbox{ \& } \cdots \nonumber\\
& \mbox{ \& } A_m < \ell(A_m) \nonumber\\
q_m: \mbox{  } &A_1 \in [t[A_1], \ell(A_1)) \mbox{ \& } \cdots \mbox{ \& } A_{m-1} \in [t[A_{m-1}, \nonumber\\
& \ell(A_{m-1})) \mbox{ \& } A_m < t[A_m] \label{mdqueries}
\end{align} 
Again, Figure~\ref{fig:2DE} shows an example of $q_1$ and $q_2$ for the 2D space.

\begin{figure}
\center
\includegraphics[width=0.35\textwidth]{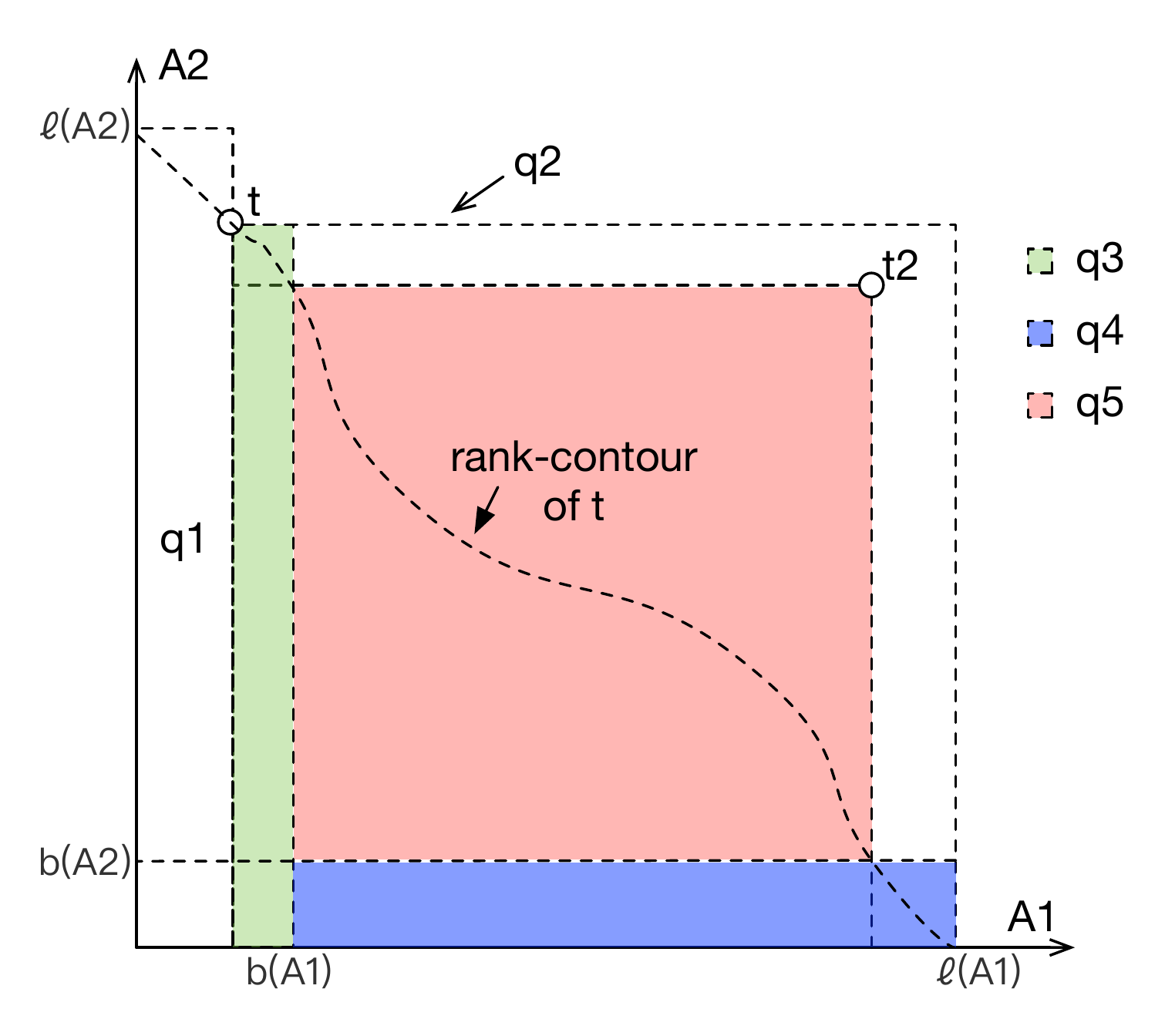}
\caption{Example of MD-BASELINE}
\label{fig:2DE}
\end{figure}

One can see that the union of these $m$ (mutually exclusive) queries covers in its entirety the region ``underneath'' the rank-contour of $t$. Thus, if none of them overflows, we can safely conclude that the No.~1 tuple must be either $t$ or one of the tuples returned by one of the $m$ queries. If at least one query overflows and returns $t^\prime$ with score $\mathcal{S}(t^\prime) < \mathcal{S}(t)$, i.e., $t^\prime$ that ranks higher than $t$, we restart the entire process with $t = t^\prime$.

Otherwise, for each query $q_i$ that overflows, we ``partition'' it further into $m + 1$ queries.  Let $t_i$ be the tuple returned by $q_i$.  We compute for each attribute $A_j$ a value $b(A_j)$ such that
\begin{align}
b(A_j) = \min\{v \in V(A_j) | \mathcal{S}(t) \leq \nonumber\\
\mathcal{S}(t_i[A_1], \ldots, t_i[A_{j-1}], b(A_j), t_i[A_{j+1}], \ldots, t_i[A_m])\}.
\end{align}
Intuitively, $b(A_j)$ can be understood as follows: In order for a tuple $t^\prime$ to outrank $t$, the highest-ranked tuple discovered so far, it must either ``outperform'' $b(A_j)$ on at least one attribute, i.e., $\exists A_j$ with $t^\prime[A_j] < b(A_j)$, or it must dominate $t_i$. Examples of $b(A_1)$ and $b(A_2)$ are shown in Figure~\ref{fig:2DE}.

Note that, while any monotonic (user-defined) ranking function yields a unique solution for $b(A_j)$, the complexity of computing it can vary significantly depending upon the design of the ranking function.  Nonetheless, recall from \S~\ref{sec2} that our main efficiency concern is on the {\em query cost} of the reranking process rather than the computational cost for solving $b(A_j)$ locally (which does not incur any additional query to the underlying database). Furthermore, the most extensively studied ranking function in the literature, a linear combination of multiple attributes, features a constant-time solver for $b(A_j)$.

Given $b(A_j)$, we are now ready to construct the $m + 1$ queries we issue. The first $m$ queries $q_{i1}, \ldots, q_{im}$ cover those tuples outperforming $b(A_1), \ldots, b(A_m)$ on $A_1, \ldots, A_m$, respectively; while the last one covers those tuples dominating $t_i$. Specifically, $q_{ij}$ ($j \in [1, m]$) is the AND of $q_i$ and
\begin{align}
(A_1 \geq b(A_1)) \mbox{ AND } \cdots \mbox{ AND } (A_{j-1} \geq b(A_{j-1})) \nonumber\\
\mbox{ AND } (A_j < b(A_j))
\end{align}
The last query is the AND of $q_i$ and $A_1 \leq t_i[A_1]$ AND $\cdots$ AND $A_m \leq t_i[A_m]$, i.e., covering the space dominating $t_i$.

Once again, at anytime during the process if a query returns $t^\prime$ with $\mathcal{S}(t^\prime) < \mathcal{S}(t)$, we restart the entire process with $t = t^\prime$. Otherwise, for each query that overflows, we ``partition'' it into $m + 1$ queries as described above.

In terms of query cost, recall from \S~\ref{sec2} our idea of leveraging the query history by checking if any previously discovered tuples match the query we are about to issue.  Given the idea, each tuple will be retrieved at most once by MD-BASELINE. Since each tuple we discover triggers at most $m + 1$ queries which are mutually exclusive with each other, one can see that the worst-case query cost of MD-BASELINE for discovering the top-1 tuple is $O(m \cdot n)$.

\subsubsection{Discovery of Top-$k$}

We now discuss how to discover the top-$k$ ($k > 1$) tuples satisfying a given query. To start, consider the discovery of No.~2 tuple after finding the top-1 tuple $t_1$.  What we can do is to pick an arbitrary attribute, say $A_1$, and partition the search space into two parts: $A_1 < t_1[A_1]$ and $A_1 > t_1[A_1]$.  Then, we launch the top-1 discovery algorithm on each subspace.  Note that during the discovery, we can reuse the historic query answers - e.g., by starting from the tuple(s) we have also retrieved in each subspace that have the smallest $\mathcal{S}(\cdot)$. One can see that one of the two discovered top-1s must be the actual No.~2 tuple $t_2$ of the entire space.

Once $t_2$ is discovered, in order to discover the No.~3 tuple, we only need to further split the subspace from which we just discovered $t_2$ (into two parts).  For example, if we discovered $t_2$ from $A_1 > t_1[A_1]$, then we can split it again into $A_1 \in (t_1[A_1], t_2[A_1])$ and $A_2 > t_2[A_1]$. One can see that the No.~3 tuple must be either the top-1 of one of the two parts or the top-1 of $A_1 < t_1[A_1]$, which we have already discovered. As such, the discovery of each tuple in top-$k$, say No.~$h$, requires launching the top-1 discovery algorithm {\em exactly twice}, over the two newly split subspaces of the subspace from which the No.~$h-1$ tuple was discovered. Thus, the worst-case query cost for MD-BASELINE to discover all top-$k$ tuples is $O(m \cdot n \cdot k)$.

\subsection{MD-Binary}

\subsubsection{Problem of MD-Baseline}

A main problem of MD-Baseline is its poor performance when the system ranking function is {\em negatively correlated} to the user-desired ranking function. To understand why, consider how MD-Baseline compared with the 1D-Baseline algorithm discussed in \S~\ref{sec3}. Both algorithms are iterative in nature; and the objectives for each iteration are almost identical in both algorithms: once a tuple $t$ is discovered, find another tuple $t^\prime$ that outranks it according to the input ranking function. The difference, however, is that while it is easy to construct in 1D-Baseline a query that covers only those tuples which outranks $t$ (for the attribute under consideration), doing so in the MD case is impossible.

\begin{figure}
\center
\includegraphics[width=0.35\textwidth]{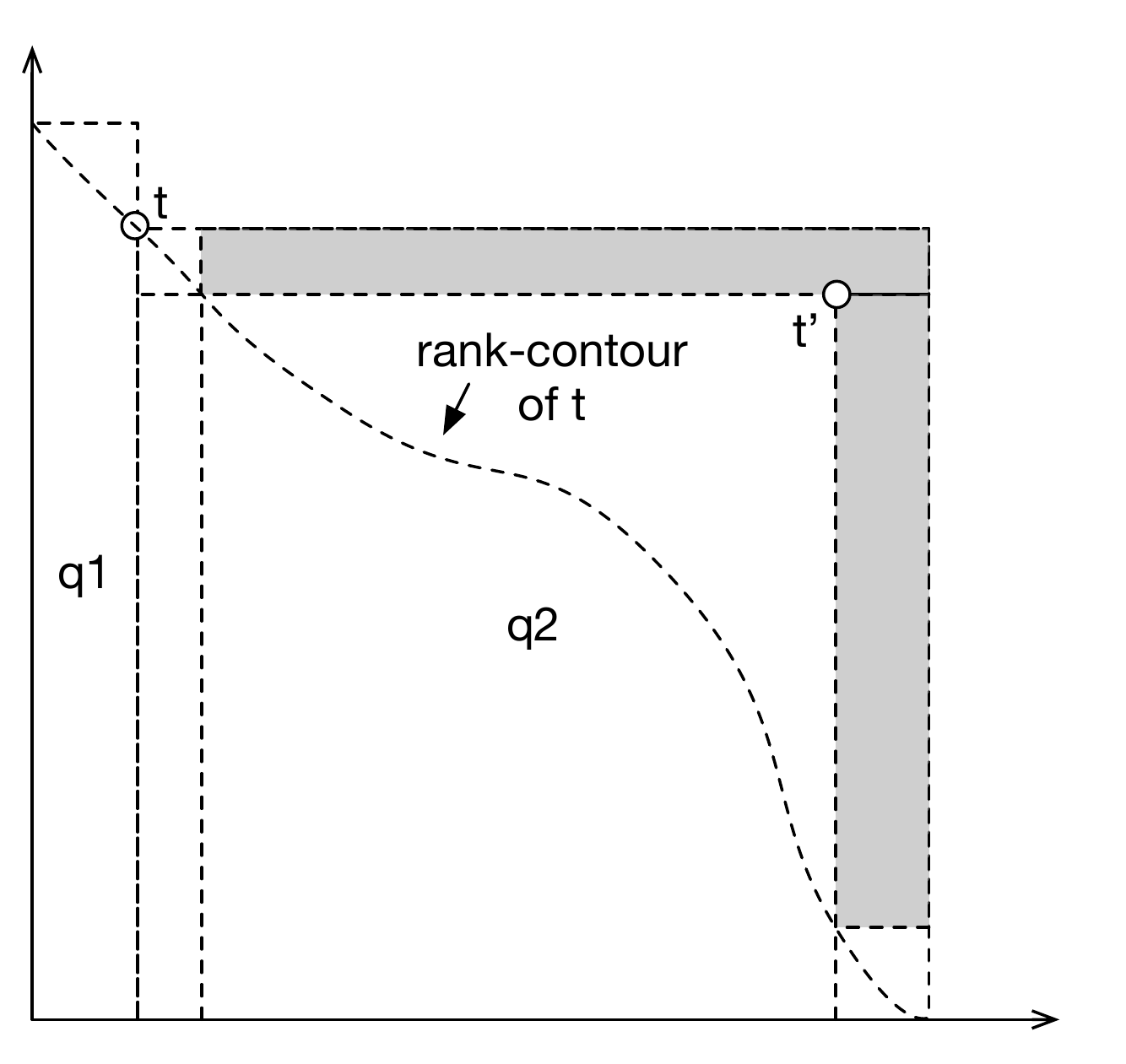}
\caption{Illustration of problem with MD-Baseline}
\label{fig:2D1}
\end{figure}

The reason for this difference is straightforward: observe from Figure~\ref{fig:2D1} that, when there are more than one, say two, attributes, the subspace of tuples outranking $t$ is roughly ``triangular'' in shape. On the other hand, only ``rectangular'' queries are supported by the database. This forces us to issue at least $m$ queries to ``cover'' the subspace outranking $t$ (without covering, and returning, $t$ itself).

The problem for this ``coverage'' strategy in MD-Baseline, however, is that the rectangular queries it issues may match many tuples that indeed rank lower (i.e., have larger $\mathcal{S}(\cdot)$)) than $t$ according to the desired ranking function. For example, half of the space covered by $q_2$ in Figure~\ref{fig:2D1} is occupied by tuples that rank lower than $t$. This means that, when the system ranking function is negatively correlated with our desired one, queries like $q_2$ in Figure~\ref{fig:2D1} are most likely going to return tuples that rank lower than $t$. This outcome has two important ramifications on the efficiency of MD-Baseline: First, it significantly slows down the process of iteratively finding a tuple that outranks the previous one. Second, within each iteration, it slows down the pruning of the search space. For example, observe from Figure~\ref{fig:2D1} that, after $q_2$ returns $t^\prime$, the pruning effect on the space covered by $q_2$ is minimal, i.e., only the dark subspace on the top-right corner of $q_2$.

\subsubsection{Design of MD-Binary}

We propose two ideas in MD-Binary to address the two ramifications of MD-Baseline, respectively:

\vspace{1mm}
\noindent {\bf Direct Domination Detection:} The intuition of this idea can be stated as follows: When a query such as $q_2$ returns a tuple $t^\prime$ that ranks lower than $t$, we attempt to ``test'' whether this is indeed caused by the absence of higher-ranked tuples in $q_2$, or by the ill-conditioned nature of the system ranking function. As discussed above, there is no way to efficiently cover the subspace of tuples outranking $t$. Thus, what we do here is to find the single query which (1) is a subquery of $q_2$, (2) only covers the subspace outranking $t$, and (3) has the maximum volume among all queries that satisfy (1) and (2).

\begin{figure}
\center
\includegraphics[width=0.35\textwidth]{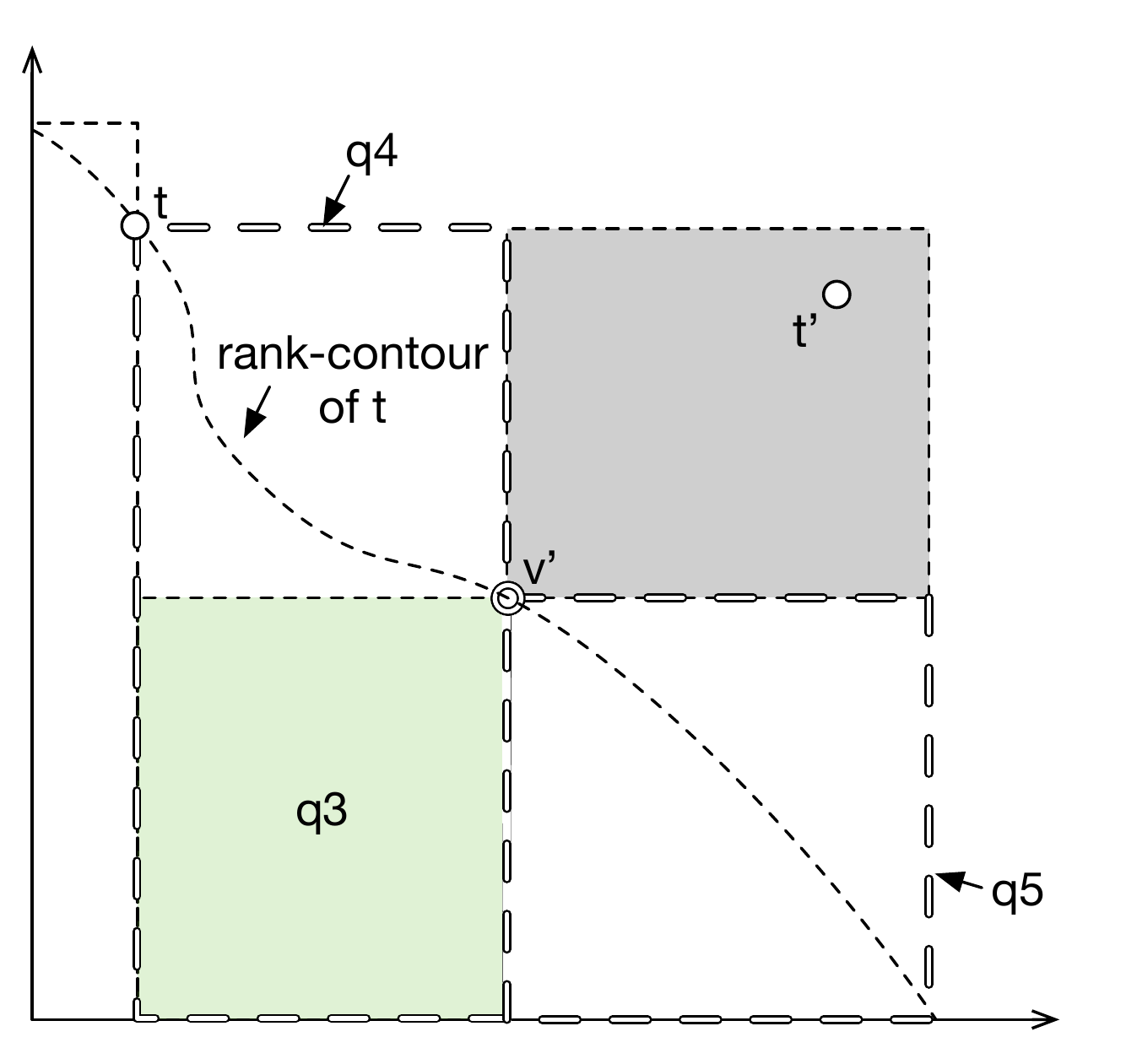}
\caption{Design of MD-Binary: Example 1}
\label{fig:2D2}
\end{figure}

For example, when $q_2$ in Figure~\ref{fig:2D1} returns $t^\prime$, we issue $q_3$ (marked in green) in Figure~\ref{fig:2D2} which covers roughly half of the ``triangular'' subspace underneath the rank-contour of $t$ in $q_2$. As another example, if $q_1$ in Figure~\ref{fig:2DE} returns a tuple with lower rank than $t$, then we the max-volume tuple would be $q_7$ in Figure~\ref{fig:2D3}, which covers almost all of the subspace outranking $t$ in $q_1$. One can see from these examples that, if the returning of $t^\prime$ is caused by the ill-conditioned system ranking function while there are abundant tuples outranking $t$, then $q_3$ and/or $q_7$ are likely to return such a tuple and successfully push MD-Binary to the next iteration. If, on the other hand, $q_3$ returns empty, we use the next idea to further partition $q_2$, in order to determine whether there is any tuple in it that outranks $t$.

\vspace{1mm}
\noindent {\bf Virtual Tuple Pruning:} We now address the second problem of MD-Baseline, i.e., the lack of pruning power when the system ranking function is negatively correlated with the desired one. To this end, our idea is to prune the search space according to {\em not} the returned tuple, but a {\em virtual tuple} created for the purpose of minimizing the pruned subspace. Figure~\ref{fig:2D2} illustrates an example: Instead of partitioning $q_2$ with $t^\prime$ like in Figure~\ref{fig:2D1} which results in minimal pruning, we ``create'' a virtual tuple $v^\prime$ which maximizes the reduction of search space as marked in gray in Figure~\ref{fig:2D2}.

\begin{figure}
\center
\includegraphics[width=0.35\textwidth]{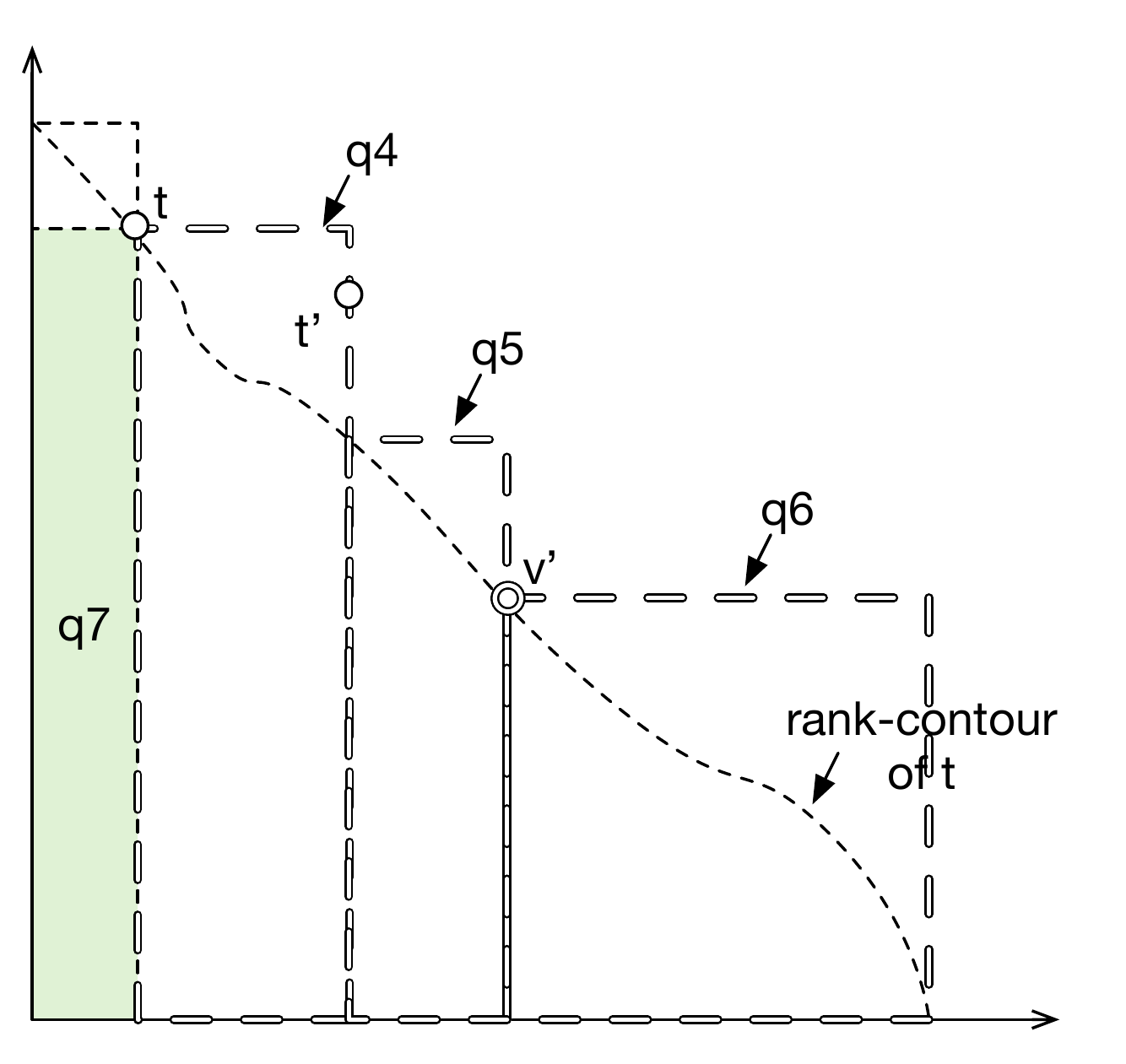}
\caption{Design of MD-Binary: Example 2}
\label{fig:2D3}
\end{figure}

Figure~\ref{fig:2D2} represents one possible outcome of virtual tuple pruning, when $v^\prime$ happens to dominate the tuple $t^\prime$ returned by $q_2$.  The other possible outcome is depicted in Figure~\ref{fig:2D3}, where $v^\prime$ does not dominate $t^\prime$. In this case, if we still split $q_2$ as in Figure~\ref{fig:2D2}, then one of the subspace (i.e., $x \in (t[x], v^\prime[x])$ AND $y < t[y]$) would return $t^\prime$, making the query answer useless. As such, we split $q_2$ into three pieces in this scenario, as shown in Figure~\ref{fig:2D3}.

The more general design of virtual tuple pruning for an $m$-D database is shown in Algorithm~\ref{alg:mdbinary}. The algorithm also depicts the direct domination detection idea.  Note from the algorithm that, depending on the values of $t^\prime$ and $v^\prime$ on the $m$ attributes, the number of split subspaces can range from $m$, when $v^\prime$ dominates $t^\prime$, to $2m - 1$, when $t^\prime$ dominates $v^\prime$ on all but one attribute.

\begin{algorithm}[!htb]
\caption{{\bf MD-BINARY}}
\begin{algorithmic}[1]
\label{alg:mdbinary}
\STATE apply 1D-RERANK on $A_1$ to $t$ and set threshold=$s(t)$
\STATE add the queries in Equation~\ref{mdqueries} to the empty queue \label{line2}
\STATE {\bf while} queue is not empty
    \STATE \hindent $q\prime$=queue.delete
    \STATE \hindent $T$ = Top-$k$($q\prime$);$t$ = argmin$_{s(t)}\{t\in T\}$
    \STATE \hindent {\bf if} $s(t)<$threshold
    \STATE \hindent \hindent threshold=$s(t)$; {\bf goto} Line~\ref{line2}
    \STATE \hindent {\bf if} $T$ is valid: {\bf continue}
    \STATE \hindent $v^\prime=$ argmax$_{vol(v)}\{v\in $ contour$(t)\}$
    \STATE \hindent $T$ = Top-$k$($\forall A \in \mathcal{A}$, $A\leq v^\prime[A])$
    \STATE \hindent {\bf if} $T$ is not underflow
    \STATE \hindent \hindent $t$=argmin$_{s(t)}\{t\in T\}$; threshold=$s(t)$; {\bf goto} Line~\ref{line2} 
    \STATE \hindent {\bf for each} $A_i\in \mathcal{A}$
    \STATE \hindent \hindent {\bf if} $t[A_i]\geq v^\prime [A_i]$ add the following query to the queue
    \begin{align}
    q_1: \mbox{  } q\prime \mbox{ AND } A_i<v^\prime[A_i] \mbox{ AND } \{\forall_{j=1} ^{i-1} A_j>=v^\prime[A_j]\} \nonumber
	\end{align}
    \STATE \hindent \hindent {\bf else} add the following queries to the queue
    \begin{align}
	q_1: \mbox{  } q\prime \mbox{ AND } A_i<t[A_i] \mbox{ AND } \{\forall_{j=1} ^{i-1} A_j>=v^\prime[A_j]\} \nonumber \\
	q_2: \mbox{  } q\prime \mbox{ AND } A_i<v^\prime[A_i] \mbox{ AND } A_{i+1}<b_t(A_{i+1}) \nonumber \\
	 \mbox{ AND } \{\forall_{j=1} ^{i-1} A_j>=v^\prime[A_j]\} \nonumber
	\end{align}
\STATE {\bf return} $t$
\end{algorithmic}
\end{algorithm}

One can see from the design that virtual tuple pruning does not affect the correctness of the algorithm: so long as $\mathcal{S}(v^\prime) \geq \mathcal{S}(t)$, the union of the split subspaces still cover $q_2$. On the other hand, the benefit of the idea can be readily observed from Figure~\ref{fig:2D2}: instead of having only a small reduction of the search space like in Figure~\ref{fig:2D1}, now we can prune half of the space in $q_2$ that rank below $t$ (in this 2D case, of course).  The experimental results in \S~\ref{sec:experiments} demonstrate the effectiveness of virtual tuple pruning.% over real-world databases.

\subsection{MD-RERANK}

Just like the 1D case, the query cost of MD-Binary may increase significantly when there is a dense cluster of tuples right above the rank-contour of the top-1 tuple. In this case, the split in MD-Binary may have to continue for a large number of times before all tuples in the cluster are excluded from the search space.  Once again, our solution to this problem is index-based reranking. Like in the 1D case, we proactively {\em record as an index} densely located tuples once we encounter them, so that we do not need to incur a high query cost every time a query $q$ triggers visits to the same dense region.

More specifically, MD-RERANK follows MD-Binary until a remaining search space (1) is covered by an already crawled region in the index; or (2) has volume smaller than $|V| \cdot (s/n) / c$, where $|V|$ is the volume of the entire data space, and $s$ and $c$ are the same as in 1D. In the earlier case, since the search space has been crawled already, we can directly reuse the crawled tuples. In the latter case, we follow the same procedure as in 1D-RERANK, i.e., we crawl the space and, if it indeed turns out to be dense (by containing at least $s$ tuples), we include the crawled tuples into the index.
Algorithm~\ref{alg:rerankmd} depicts the pseduocode of MD-RERANK.

\begin{algorithm}[!htb]
\caption{{\bf MD-RERANK}}
\begin{algorithmic}[1]
\label{alg:rerankmd}
\STATE follow MD-BINARY
\STATE during the process for each query $q\prime$:
\STATE \hindent {\bf if} $V(q\prime) < |V| \cdot (s/n) / c$
\STATE \hindent \hindent $q\prime$ = remove $Sel(q)\}$ from $q\prime$
\STATE \hindent \hindent {\bf if} ORACLE($q\prime$) exists
\STATE \hindent \hindent \hindent {\bf return} argmin$_{s(t)}\{t\in D(q\prime))|$ $t$ matches $Sel(q)\}$
\STATE \hindent \hindent $t$=MD-BASELINE($q\prime$); add $t$ to temp
\STATE \hindent \hindent {\bf while} $t$ does not satisfy $Sel(q)$
\STATE \hindent \hindent \hindent $t_1$ = MD-BASELINE($q\prime$ AND $A_1<t[A_1]$)
\STATE \hindent \hindent \hindent $t_2$ = MD-BASELINE($q\prime$ AND $A_1>t[A_1]$)
\STATE \hindent \hindent \hindent $t$=min($t_1$ , $t_2$); add $t_1$ and $t_2$ to temp
\STATE \hindent \hindent add temp to $D(q\prime)$
\end{algorithmic}
\end{algorithm}

\section{discussions} \label{sec5}

\noindent{\bf General Positioning Assumption:} In previous discussions, we made the general positioning assumption, i.e., each tuple has a unique value on each attribute, for the simplicity of discussions.  We now consider the removal of this assumption.   Note that the removal of this assumption for MD-RERANK is extremely simple: the only tuple(s) that can be missed by MD-RERANK are those that have the exact same value on {\em every single attribute}.  Thus, the only post-processing step required for removing the assumption is to form a fully specified query according to No.~$h$ tuple just discovered.  If more than one, say $i$, tuples are returned, they become the No.~$h$ to No.~$(h + i - 1)$ top-ranked tuples. Removing the assumption for 1D-RERANK is slightly more complex.  For example, if we are running it over attribute $A_1$, the removal of the general positioning assumption means query  SELECT * FROM D WHERE $A_1 = t[A_1]$ might overflow. In this case, our solution is to call the crawling algorithm \cite{sheng2012optimal} to discover, one at a time, tuples satisfying $A_1 = t[A_1]$, as all of these tuples have the same rank for the purpose of 1D-RERANK.

%\vspace{2mm}
%\noindent{\bf Dynamic Database:} When the underlying client-server database changes dynamically over time, the only revisions required in the design of 1D- and MD-RERANK are on the usage of the dense indices we built in the process of previous user queries. While a straightforward solution is to throw away the index and just run the vanilla 1D- or MD-BINARY algorithms, we note that a more efficient solution is to still take advantage of the index with a ``verify-before-use'' approach.

%For example, when 1D-RERANK is ready to use the index to process the remaining search space $(x, y)$ on $A_1$, we need to first verify if recent changes to database have rendered the existing index useless.  To do so, we only need to issue one query. To understand why, consider the case where the index returns $t$ as the highest ranked tuple with $A_1 \in (x, y)$ satisfying the user query $q$. We issue SELECT * FROM D WHERE $A_1 \in (x, t[A_1]]$ AND $Sel(q)$. One can see that, if the query returns nothing but $q$, we can be confident that $t$ is still the right answer.  A similar verification process can be performed over MD-RERANK. The only complication there is that the verification requires an execution of MD-BINARY with initial tuple being $t$, and thus may incur a larger query cost than the single-query verification process for 1D-RERANK.

\vspace{2mm}
\noindent{\bf Multiple/Known System Ranking Functions:} Another interesting issue arising in practice is when the client-server database offers more than one ranking functions, often times allowing ranking over a specific attribute. For example, Amazon.com offers not only a proprietary rank by ``popularity'', the design of which is unknown, but also ranking by price, which is an attribute usually involved in the user-specified ranking function. An interesting implication of such a ``public'' ranking function is that it might boost the performance of the TA-1D algorithm discussed in the beginning of \S~\ref{sec4}. Specifically, since now TA can simply use the public ranking function on the attribute instead of calling 1D-RERANK, it may have a even lower query cost than MD-RERANK when the user-desired ranking function aligns well with the system one.

\vspace{2mm}
\noindent{\bf Point Predicates:} In this paper, we focused on cases where attributes involved in the ranking function are numeric attributes that support range queries.  While this is often the case in practice (as evidenced in real-world websites such as the aforementioned Blue Nile where all attributes such as price, carat, clarity, etc., are available as range predicates), there are also cases where a ranking attribute with only a small number of domain values can only be specified as a point predicate (i.e., of the form $A_i = v$) in the database search interface.  For 1D-RERANK, this is often a blessing because it simplifies the task to querying the attribute values in the preference order (plus the crawling-based provision as in the discussion for the general positioning assumption). On the other hand, it makes MD-RERANK much more costly, because now a conjunctive query covers a much smaller space than the range case. Thus, an intuition here is to prefer the TA-1D algorithm over MD-RERANK when a large number of attributes are searchable as point predicates only.  Due to space limitations, we leave a comprehensive study of this issue to future work.

\section{Experimental Evaluation}
\label{sec:experiments}

\subsection{Experimental Setup}

In this section, we present our experimental results over a number of several real-world datasets, offline and online. We started with the offline case by testing over a real-world dataset we have already collected. Specifically, we constructed a top-$k$ web search interface over it, and then executed our algorithms through the interface.  This offline setting enabled us to not only verify the correctness of our algorithms, but also investigate how the performance of query reranking changes with various factors such as the database size, the system ranking function, settings of the system search interface, etc. We followed the offline tests with online, live, experiments over two real-world web databases, including the largest online diamond retailer and a popular auto search website. In all these experiments, we applied the extensions described in \S~\ref{sec5} to resolve the general positioning assumption which may not hold in practice.

\noindent {\bf Offline Dataset:} We used the flight on-time dataset published by the US Department of Transportation (DOT)\footnote{downloaded from \url{http://www.transtats.bts.gov/DL_SelectFields.asp?Table_ID=236&DB_Short_Name=On-Time}}. A wide range of third-party websites use this dataset to identify on-time performance of flights, routes, airports, airlines, etc. It consists of 457,013 flight records of 14 US carriers during the month of May 2015. It has 28 attributes, out of which we selected the following 8 attributes for ranking: \textit{Dep-Delay, Taxi-Out, Taxi-In, Arr-Delay-New, CRS-Elapsed-Time, Actual-Elapsed-Time, Air-Time,} and \textit{Distance}.  The domain sizes are $1988$, $180$, $180$, $1971$, $718$, $724$, $676$, and $5000$, respectively.  For the purpose of the experiments, we considered two system ranking functions: {\it 0.3 AIR-TIME + TAXI-IN} (SR1) and {\it -0.1 DISTANCE - DEP-DELAY} (SR2). In general, SR1 has a positive correlation with the user-specified ranking functions we tested, while SR2 has a negative one. We set SR1 as the default ranking function in the experiments. The value of $k$ offered by the database is set to 10 by default.

\noindent {\bf Online Experiments:}  We conducted live experiments over two real-world web-sites: 
{\it Blue Nile} (BN) and {\it Yahoo! Autos} (YA).

\noindent {\it Blue Nile}\footnote{\url{http://www.bluenile.com/diamond-search}} is the largest diamonds online retailer in the world. At the time of our experiments, its catalog had 117,641 diamonds. We considered {\it Carat, Depth, LengthWidthRatio, Price,} and {\it Table} as the ranking attributes, and {\it Clarity, Color, Cut, Fluorescence, Polish, Shape,} and {\it Symmetry} for filtering. The domains for the ranking attributes are [0.23,22.74], [0.45,0.86], [0.49,0.89], [\$220,\$4506938] and [0.75,2.75], respectively. 
BN allows multiple ranking functions - ordering based on each attributes individually as well as by the derived attribute {\it price-per-carat}.
%\textcolor{red}{What was the re-ranking function? explain why you chose it.}

\noindent {\it Yahoo! Autos} is a popular website for buying used cars\footnote{\url{https://autos.yahoo.com/used-cars/}}. We considered the 13,169 cars listed for sale within 30 miles of New York city.  We treated {\it Price, Milage} and {\it Year} as the ranking attributes, and {\it BodyStyle, DriveType, Transmission, Name} and {\it Model} as the filtering attributes. The cars had a price range between \$0 and \$50,000, mileage between  0 and 300,000, and were manufactured between 1993 and 2016.  The default ranking function is ``distance from a predefined location'' (which is not monotonic). Additionally, it supports ranking by each of the numerical attributes individually.
%\textcolor{red}{What was the re-ranking function? explain why you chose it.}

\noindent {\bf Performance Measures}: As explained in \S~\ref{sec2}, our algorithms always return the precise query answer.  After verifying the correctness in all offline experiments, we turn our attention to the key performance measure, efficiency, which is measured by the number of queries issued to the web database.

\begin{figure*}[ht]
    \begin{minipage}[t]{0.23\linewidth}
        \centering
        \includegraphics[scale=0.26]{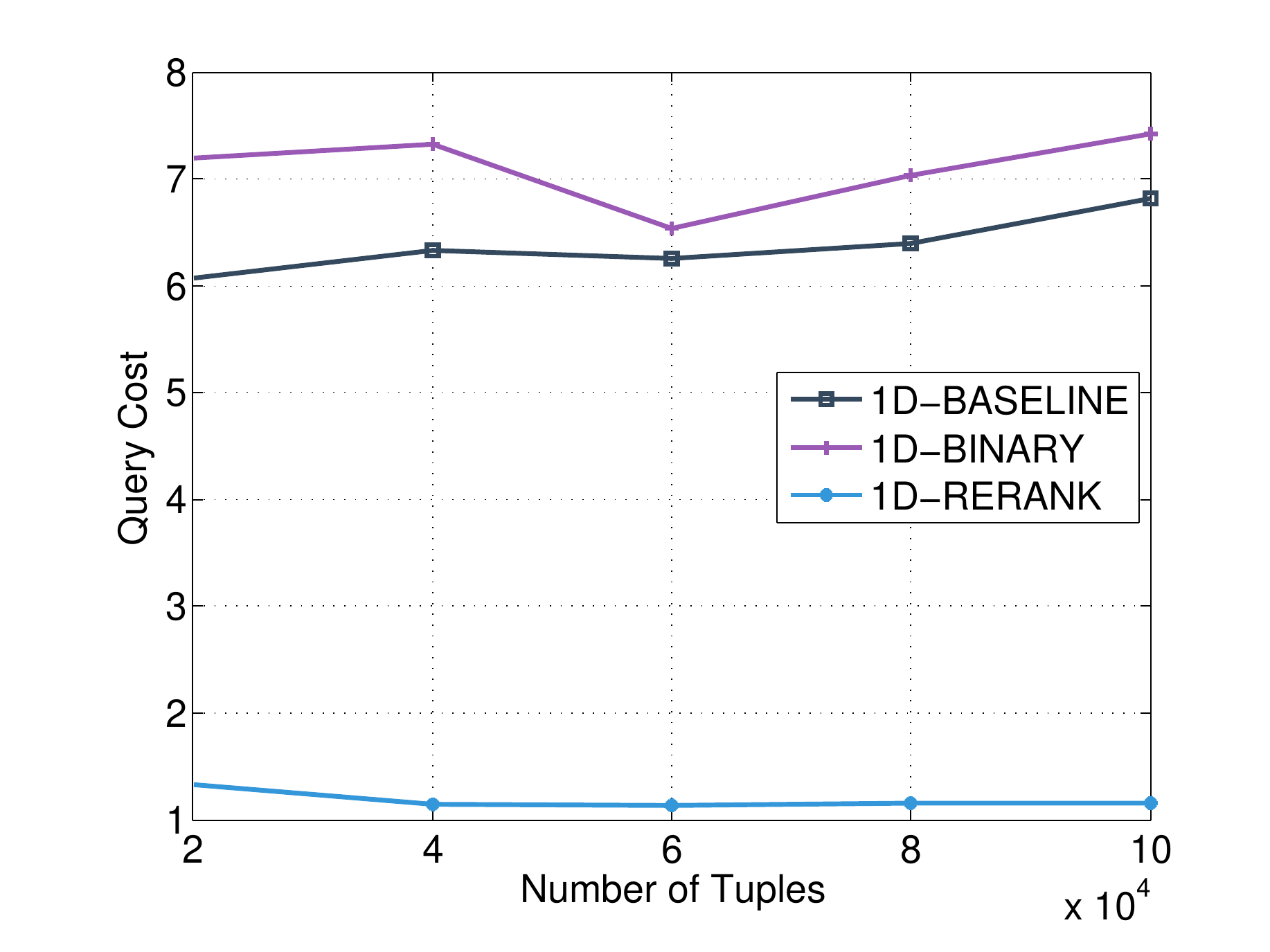}
		\vspace{-2mm}\caption{1D: Impact of $n$ (SR1)}
        \label{fig:1dn1}
    \end{minipage}
    \hspace{1mm}
    \begin{minipage}[t]{0.23\linewidth}
        \centering
        \includegraphics[scale=0.26]{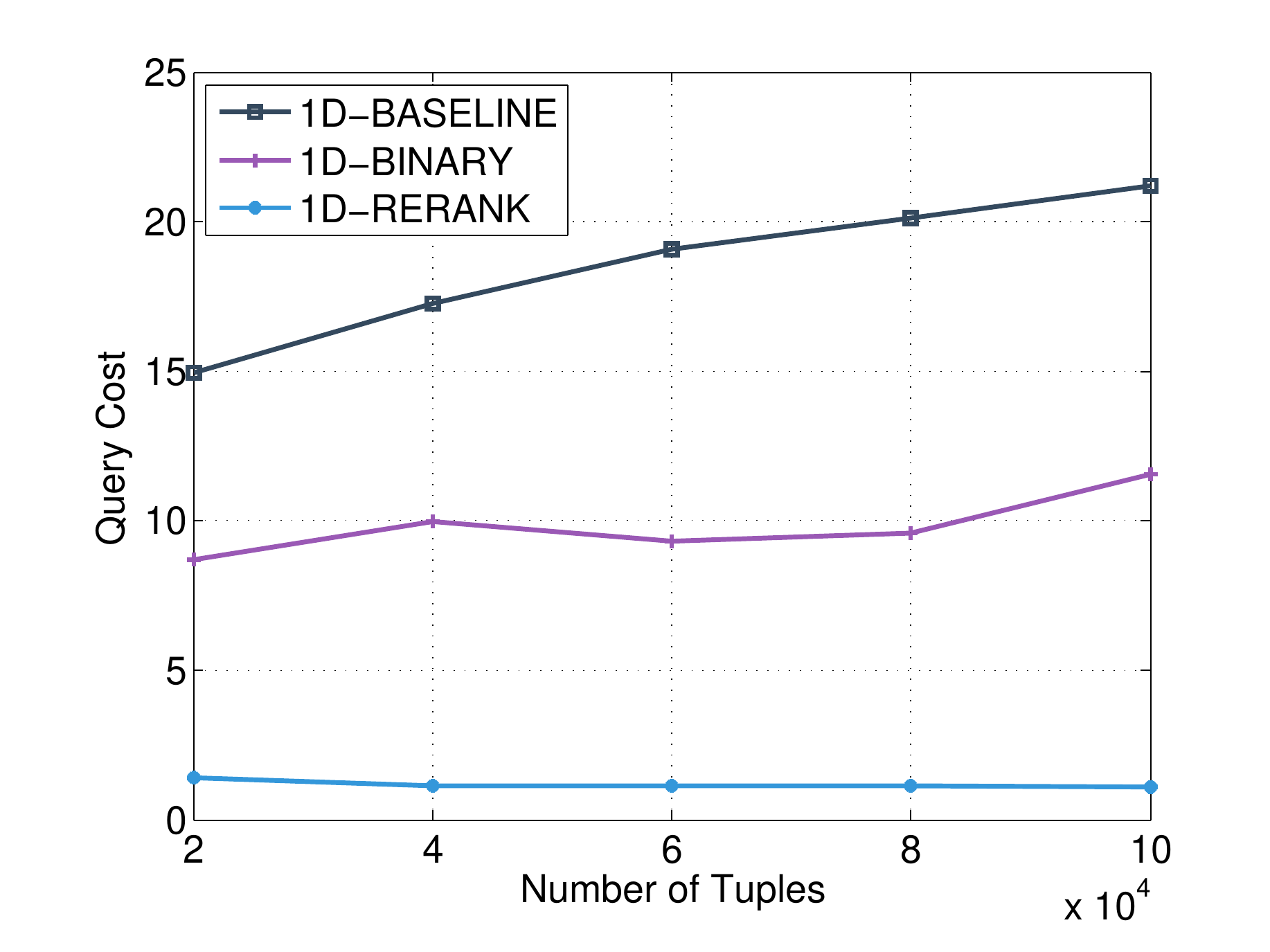}
        \vspace{-2mm}\caption{1D: Impact of $n$ (SR2)}
        \label{fig:1dn2}
    \end{minipage}
    \hspace{1mm}
    \begin{minipage}[t]{0.23\linewidth}
        \centering
        \includegraphics[scale=0.26]{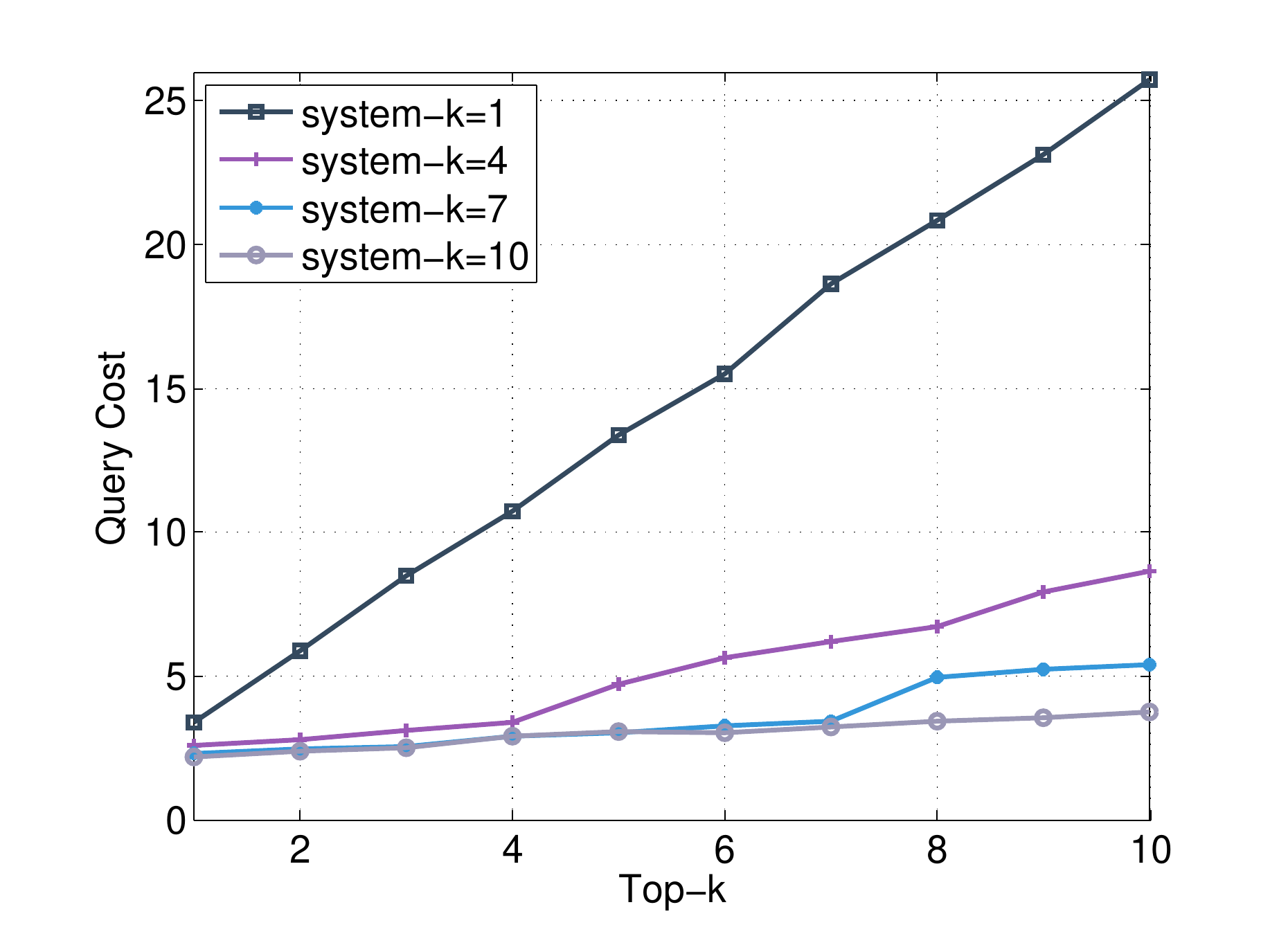}
        \vspace{-2mm}\caption{1D: Impact of System-$k$}
        \label{fig:1dk}
    \end{minipage}
    \hspace{3mm}
    \begin{minipage}[t]{0.23\linewidth}
        \centering
        \includegraphics[scale=0.26]{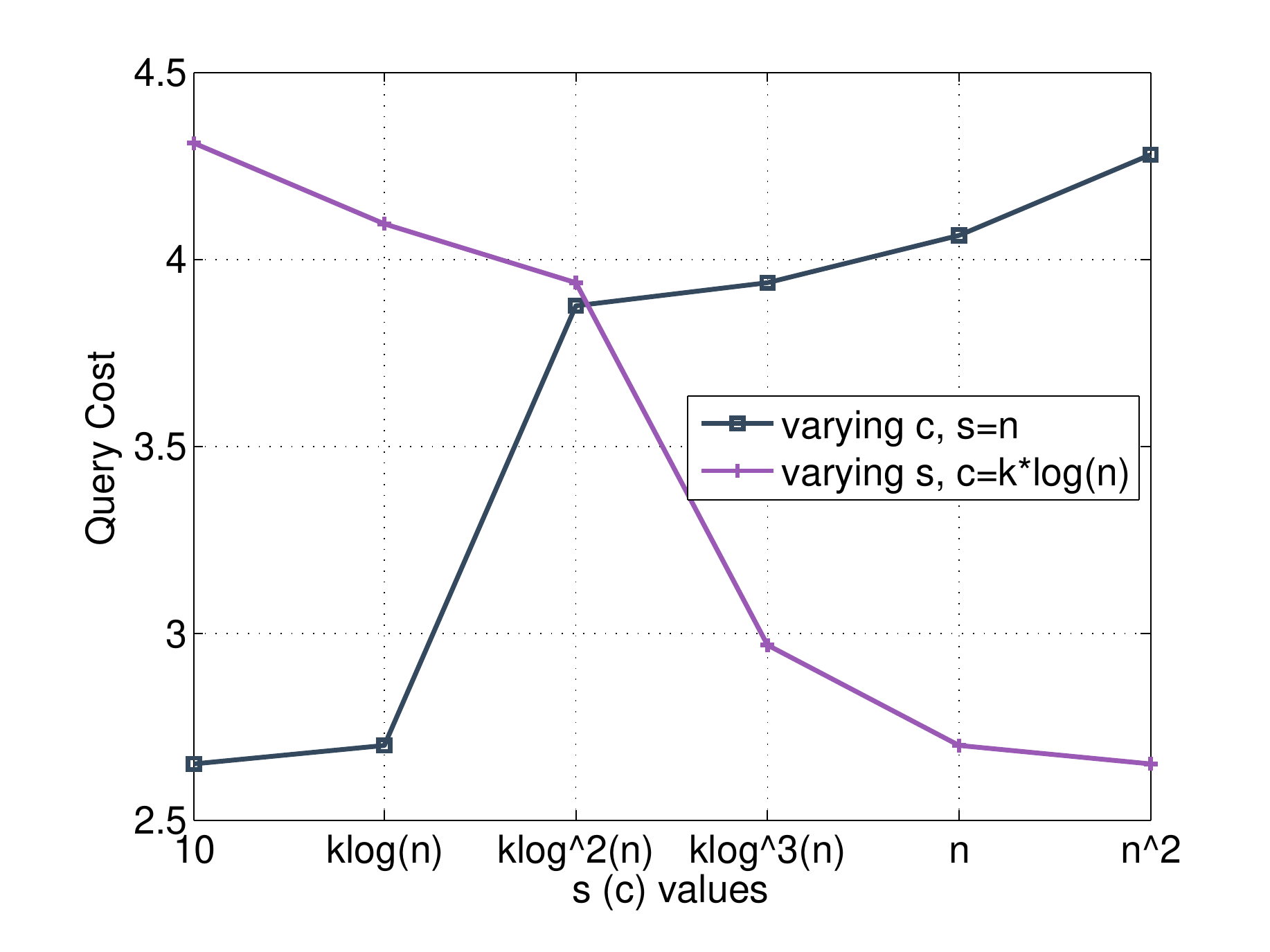}
        \vspace{-2mm}\caption{1D: Impact of $s$ and $c$}
        \label{fig:1dsc}
    \end{minipage}
    \hspace{-2mm}
\end{figure*}

\subsection{1D Experiments}
\label{sec:exp1d}
%Through out the paper we assumed general positioning on the numerical attributes, which may not be the case in real world. We notice that this assumption only affects the GetNext operation, as there may be more than one tuple having the value equal to the discovered top tuple. In order to adjust the algorithms, if the last issued query returns NULL, i.e. the query that contained the Top-$i$ was an overflow, we use the crawling algorithm proposed in \cite{sheng2012optimal} to crawl the tuples with the same value one the ranking attribute with Top-$i^{th}$.

\noindent
{\bf Constructing Workload of User Preference Queries:}
We tested a diverse set of user-specified queries of the form SELECT * FROM $D$ WHERE $Sel(q)$ ORDER BY $A_i$. Specifically, we randomly selected different subsets of filtering attributes for the WHERE clause, while choosing the (1D) ranking attribute uniformly at random.  This approach has a number of appealing properties. First, it covers diverse cases that include ideal, worst-case and typical scenarios. Second, since 1D-RERANK uses on-the-fly indexing to amortize the cost between different user-issued queries, our diverse query workload simulates a real-world scenario where the service is used by multiple users. For each experimental configuration, we execute each of the queries and report the average query cost. Specifically, for the DOT dataset, we constructed 32 queries of which 25\% do not have any filtering condition.  For BN, we constructed a set of 20 queries, of which 4 have no filtering conditions, while these values are 15 and 2 for YA, respectively. 

\subsubsection{Experiments over the Real-world Dataset}

\noindent {\bf Impact of Database Size and System Ranking Function:}
We started by testing the impact of database size on our algorithms for the two system ranking functions SR1 and SR2. To test databases of varying sizes, we drew 10 simple random samples of a given size from the DOT dataset, and measured the average query cost for the entire workload over these 10 small databases. Figures~\ref{fig:1dn1} and \ref{fig:1dn2} show the average query cost for retrieving the top-1 tuple over SR1 and SR2, respectively. As expected, the database size has negligible impact on the query cost. Also note from the figures that, consistent with our theoretical analysis, Algorithm 1D-RERANK outperformed both 1D-BASELINE and 1D-BINARY significantly. One can also note that the change in system ranking function has a major impact on the performance comparison between 1D-BASELINE and 1D-BINARY, yet has a negligible impact on that of 1D-RERANK, again consistent with our theoretical discussions.

\begin{figure*}[ht]
	\hspace{-1mm}
    \begin{minipage}[t]{0.23\linewidth}
        \centering
        \includegraphics[scale=0.26]{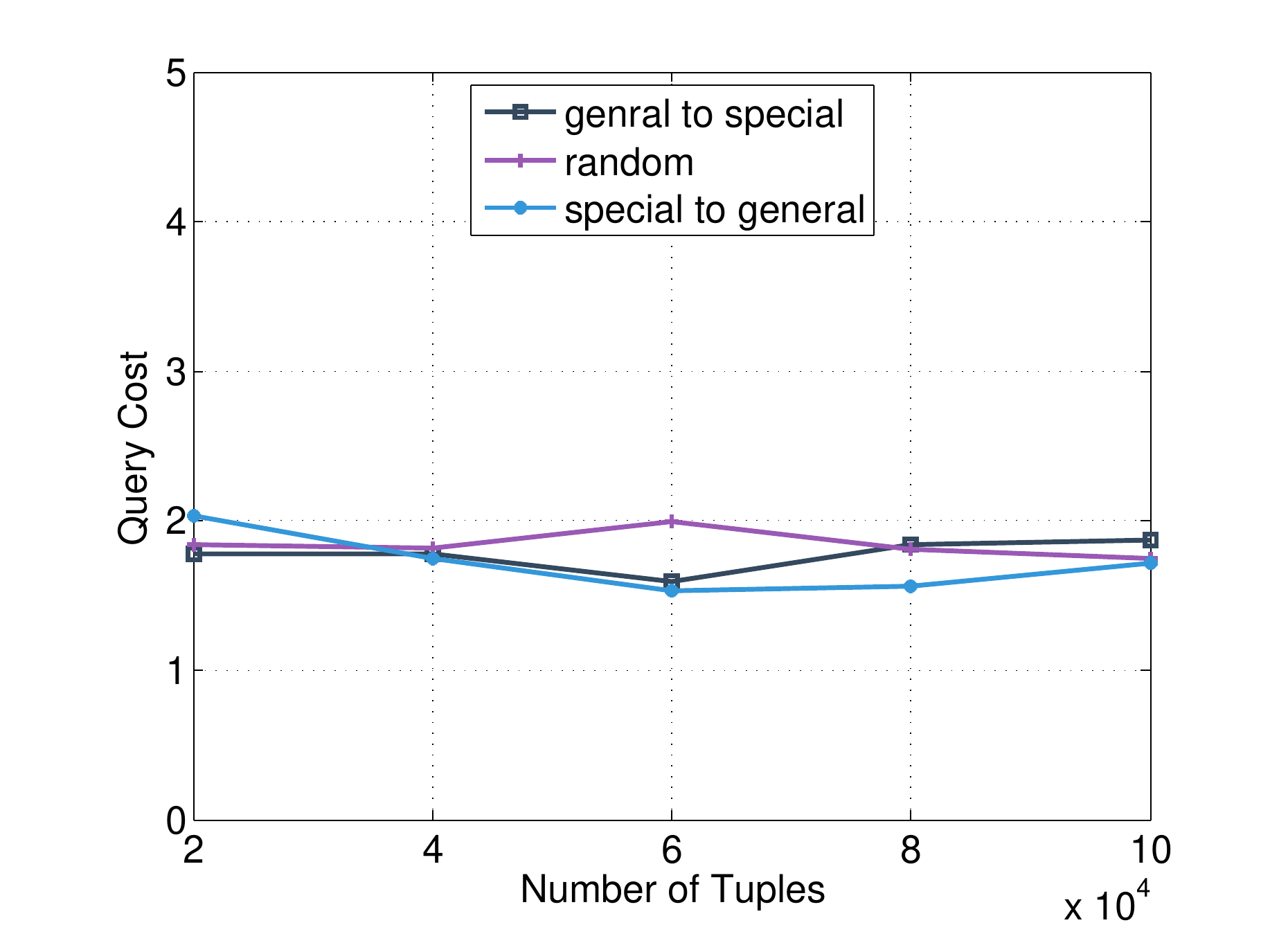}
        \vspace{-2mm}\caption{1D: Impact of Query order in 1D-RERANK}
        \label{fig:1dqo}
    \end{minipage}
    \hspace{1mm}
    \begin{minipage}[t]{0.23\linewidth}
        \centering
        \includegraphics[scale=0.26]{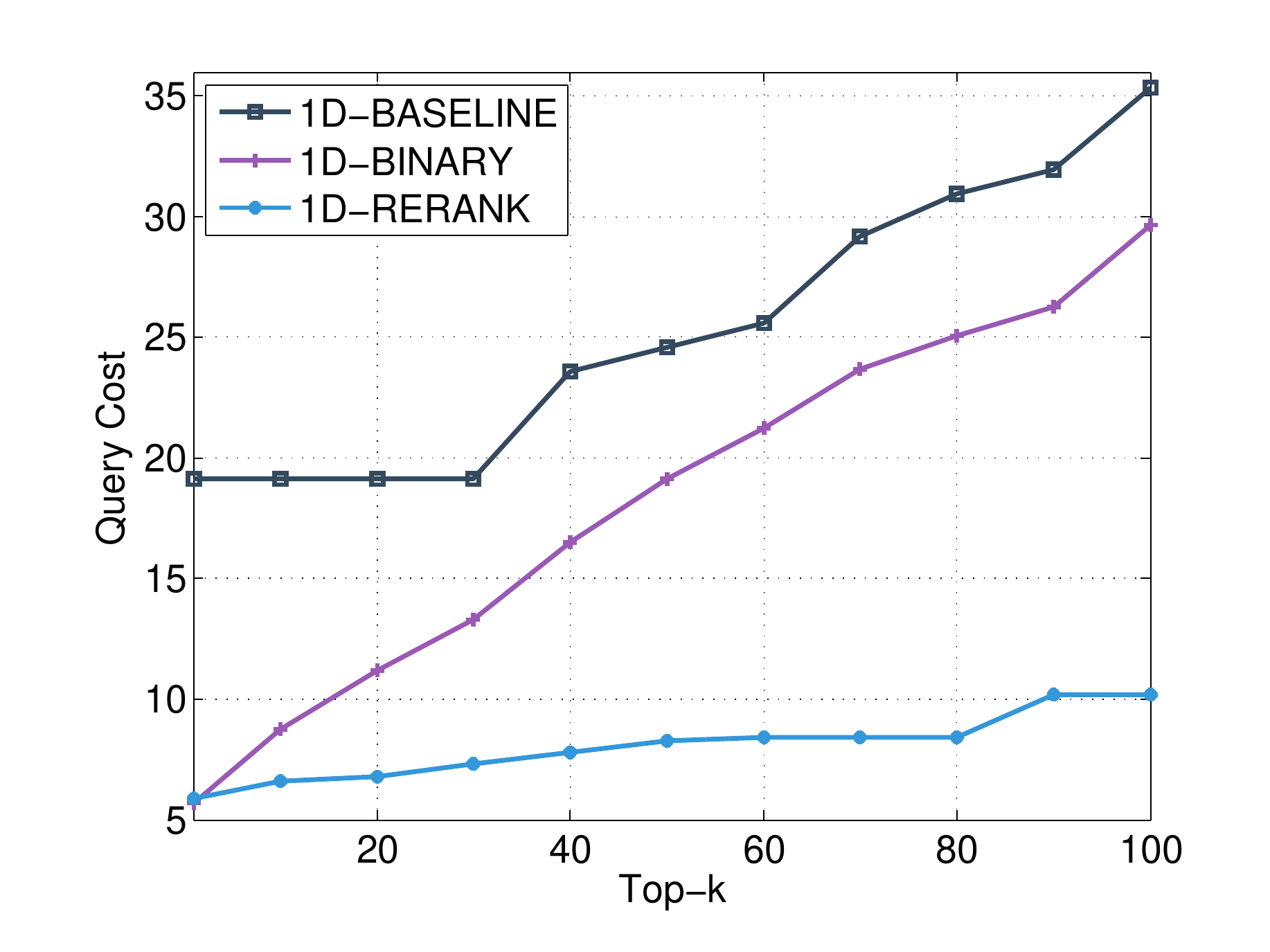}
        \vspace{-2mm}\caption{1D: Top$k$ Query Cost (BN)}
        \label{fig:1dbn}
    \end{minipage}
    \hspace{1mm}
    \begin{minipage}[t]{0.23\linewidth}
        \centering
        \includegraphics[scale=0.26]{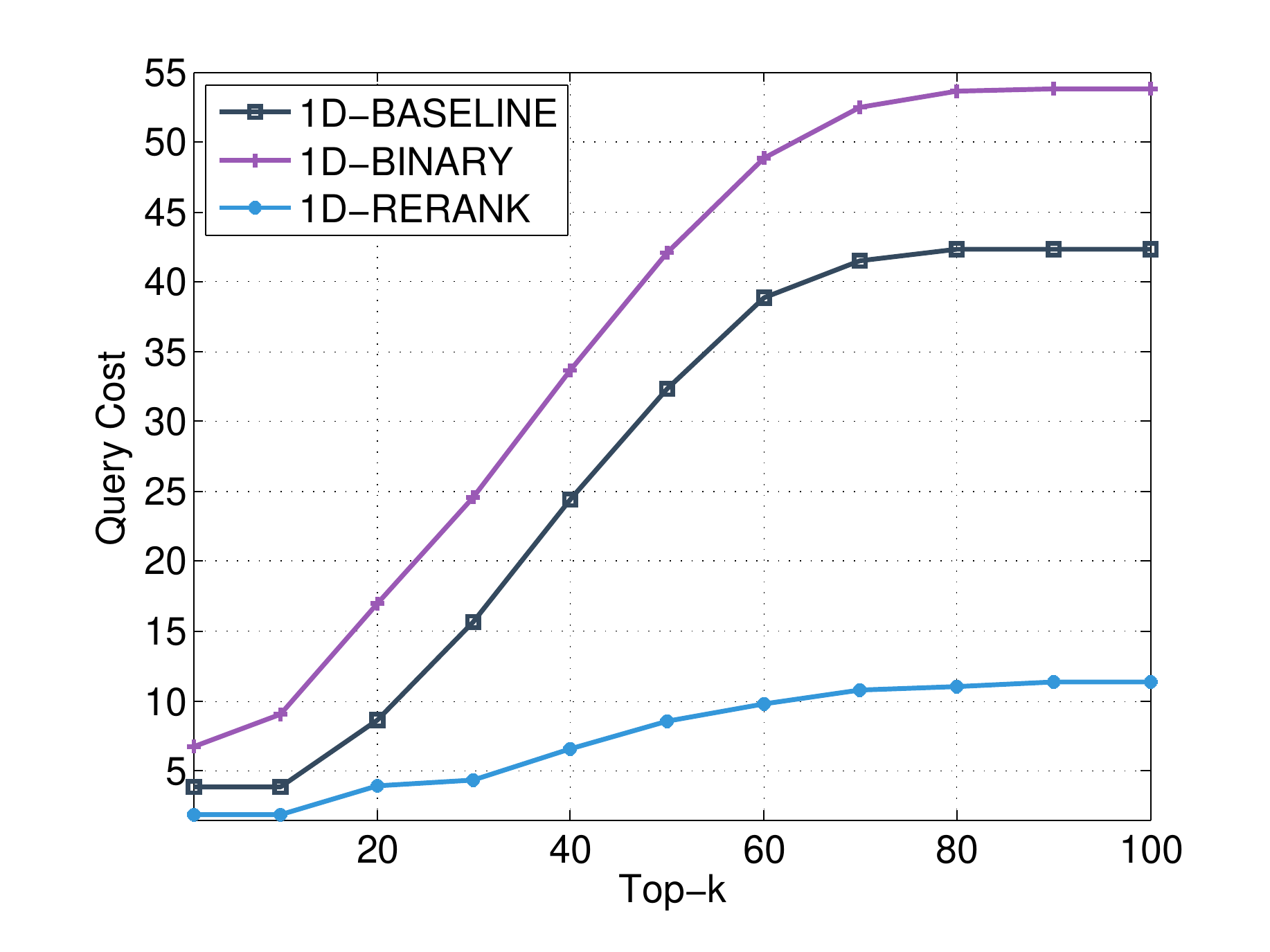}
        \vspace{-2mm}\caption{1D: Top$k$ Query Cost (YA)}
        \label{fig:1dya}
    \end{minipage}
    \hspace{3mm}
    \begin{minipage}[t]{0.23\linewidth}
        \centering
        \includegraphics[scale=0.26]{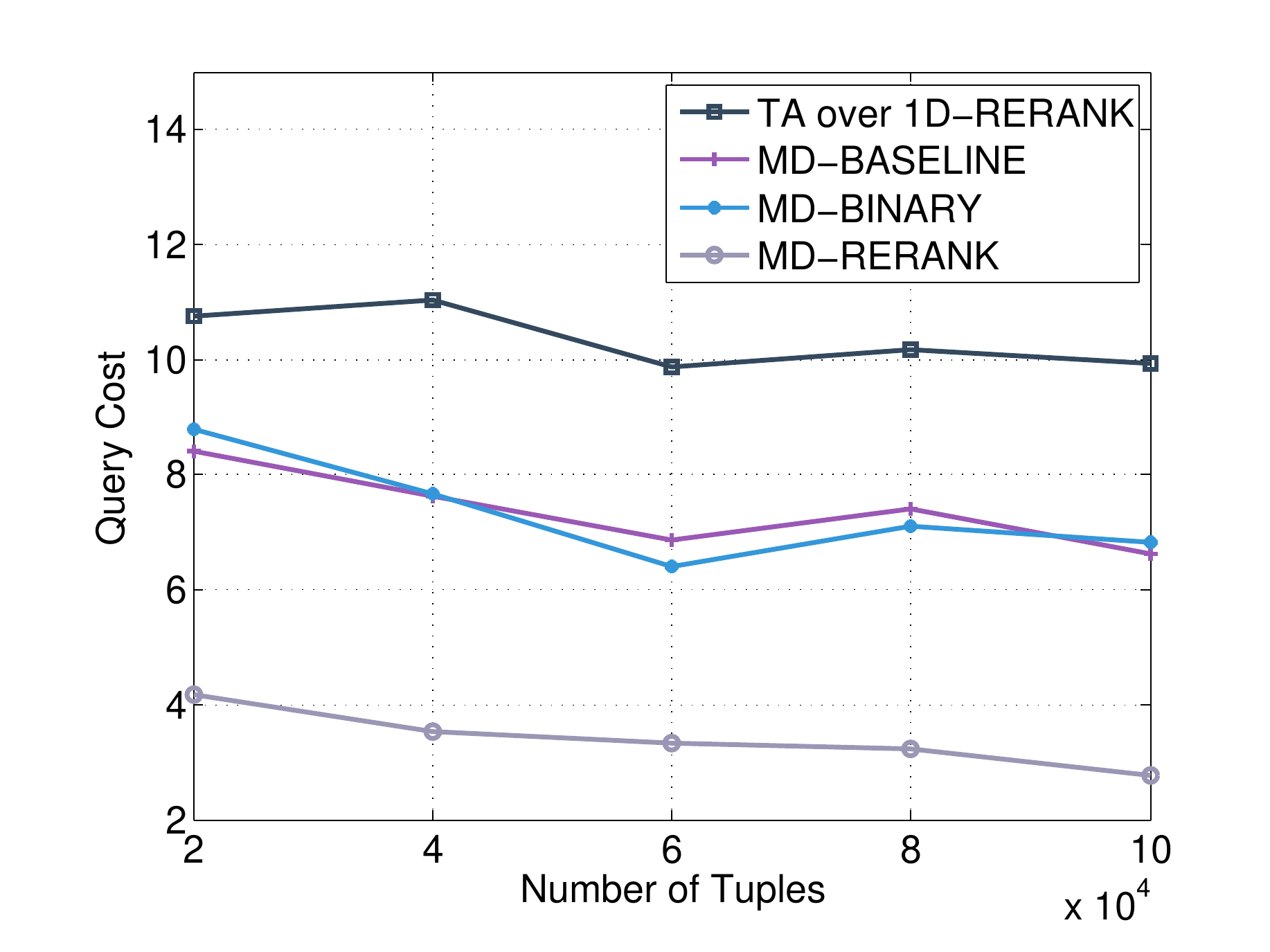}
        \vspace{-2mm}\caption{MD: Impact of $n$ (SR1)}
        \label{fig:mdn1}
    \end{minipage}
    \hspace{-2mm}
\end{figure*}

\begin{figure*}[ht]
	\hspace{-1mm}
    \begin{minipage}[t]{0.23\linewidth}
        \centering
        \includegraphics[scale=0.26]{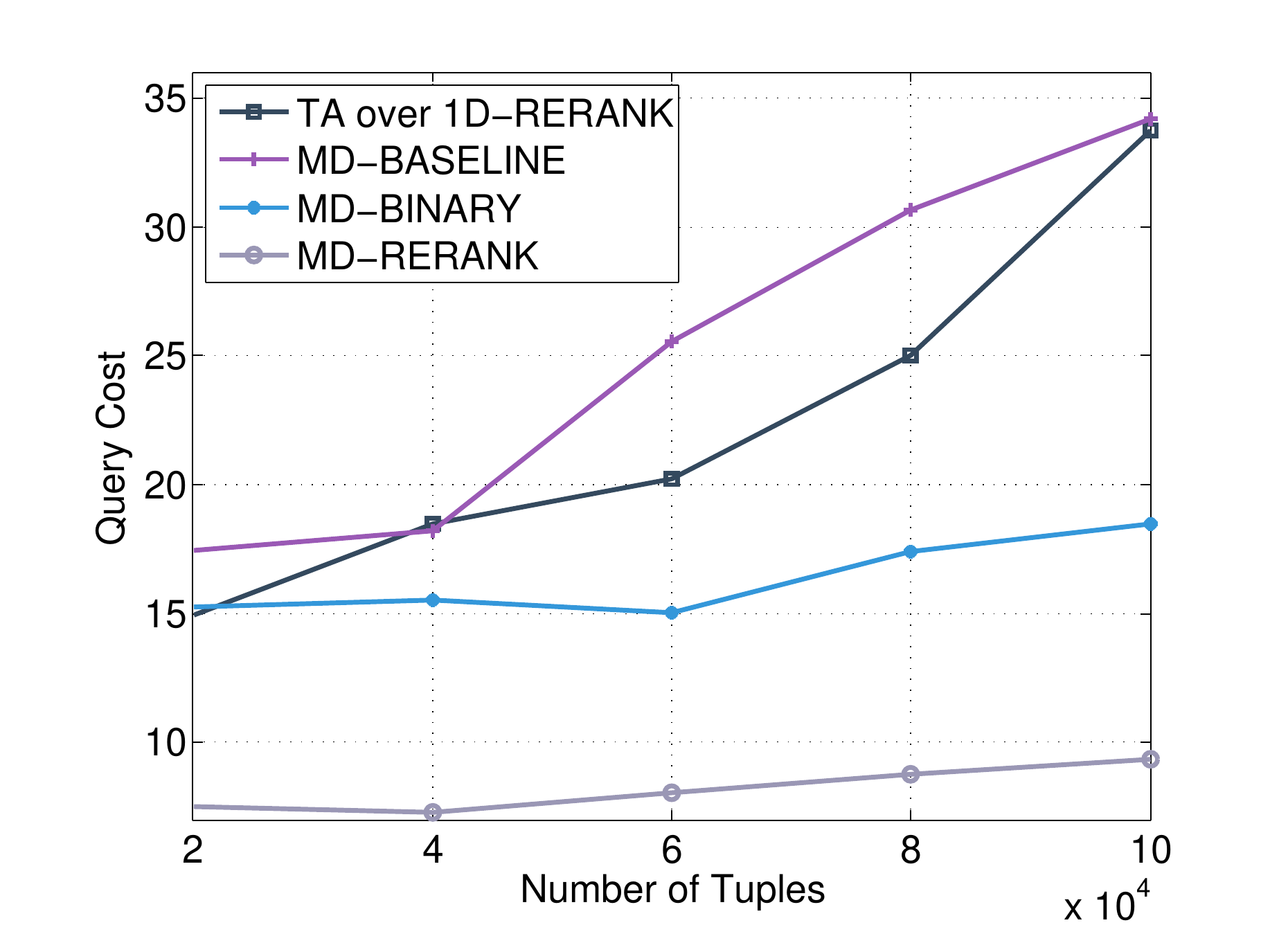}
        \vspace{-2mm}\caption{MD: Impact of $n$ (SR2)}
        \label{fig:mdn2}
    \end{minipage}
    \hspace{1mm}
    \begin{minipage}[t]{0.23\linewidth}
        \centering
		\includegraphics[scale=0.26]{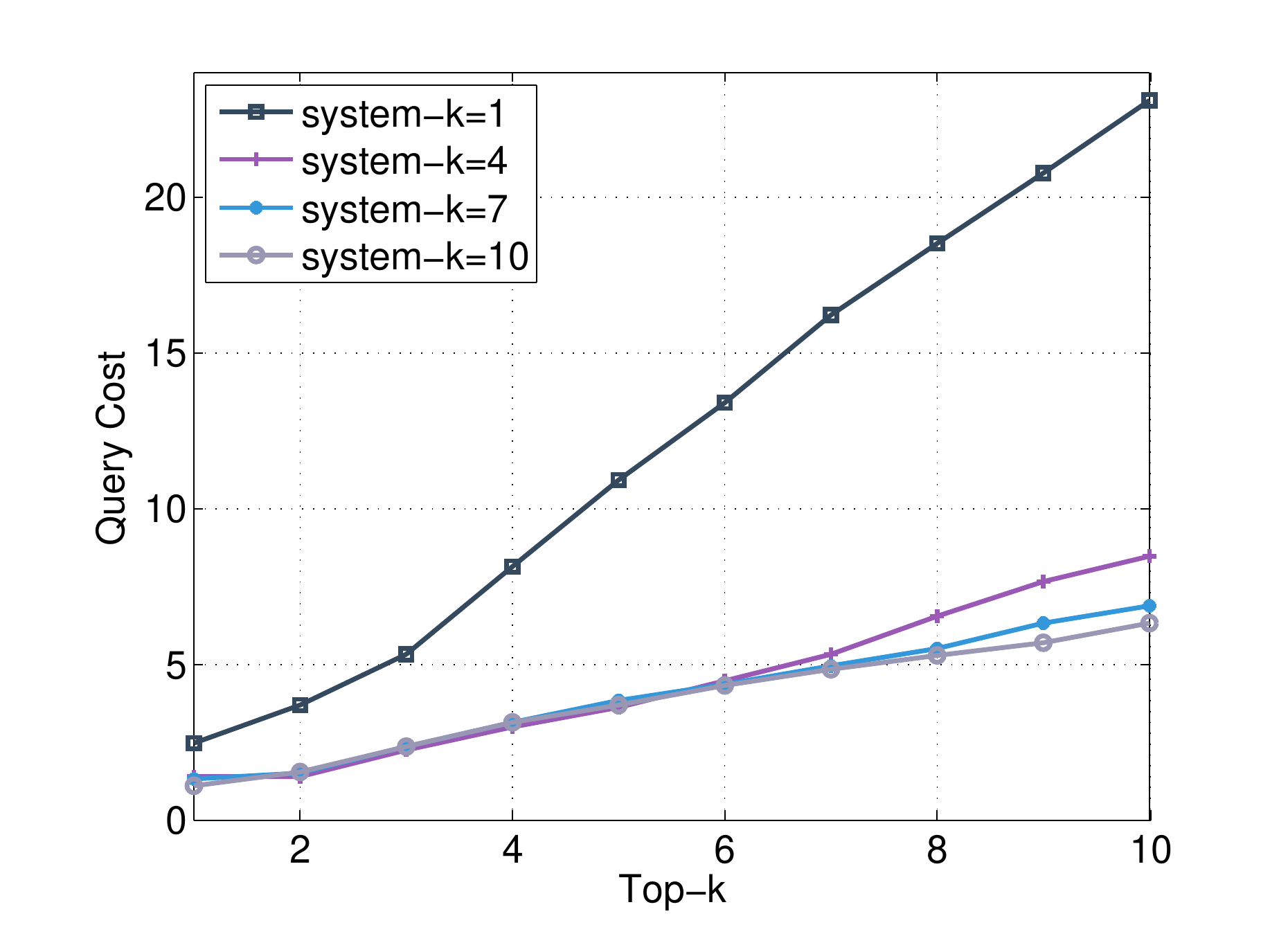}
        \vspace{-2mm}\caption{MD: Impact of System-$k$}
        \label{fig:mdk}
    \end{minipage}
    \hspace{1mm}
    \begin{minipage}[t]{0.23\linewidth}
        \centering
        \includegraphics[scale=0.26]{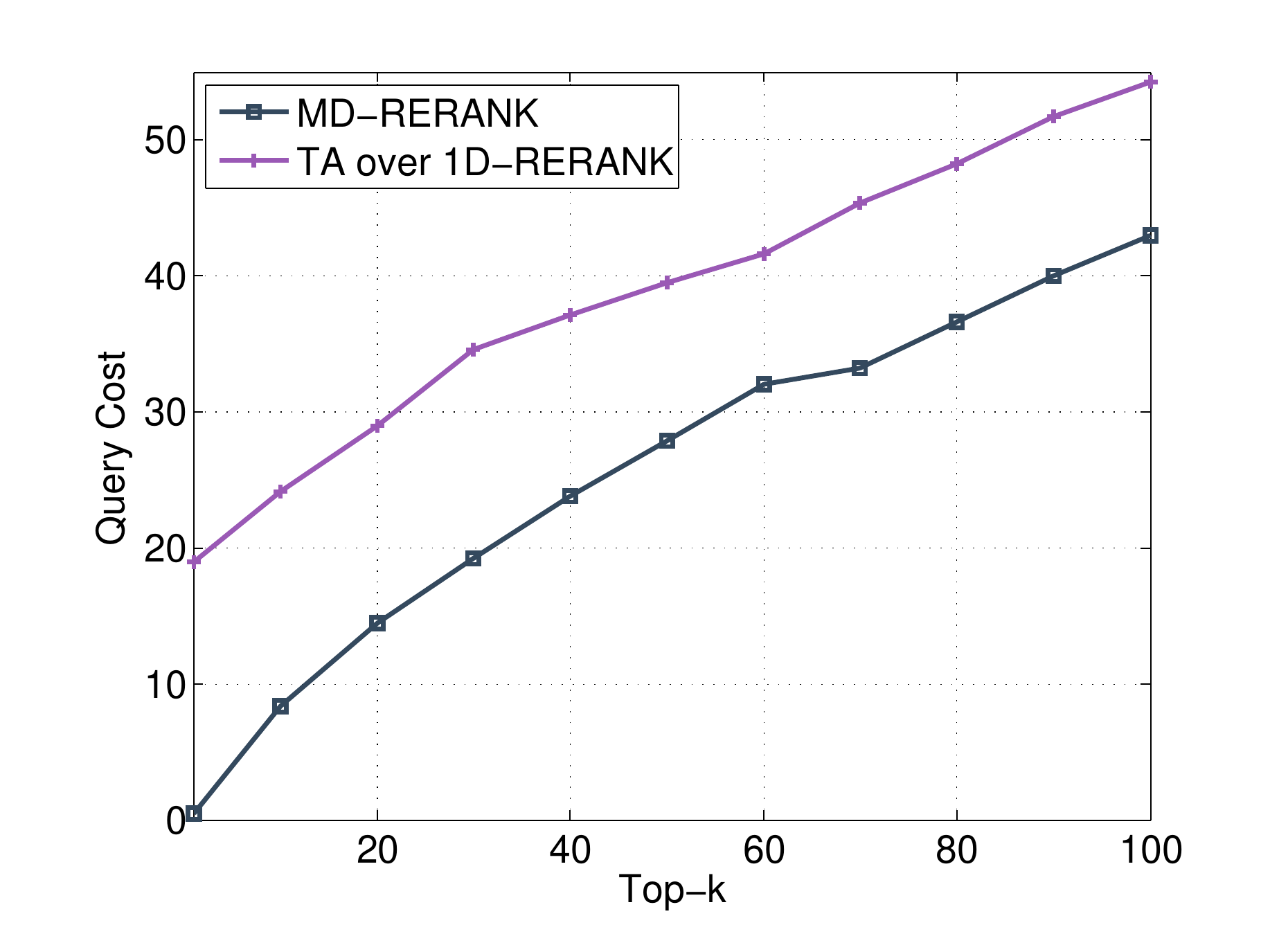}
        \vspace{-2mm}\caption{MD: Top$k$ Query Cost (BN)}
        \label{fig:mdbn}
    \end{minipage}
    \hspace{3mm}
    \begin{minipage}[t]{0.23\linewidth}
        \centering
        \includegraphics[scale=0.26]{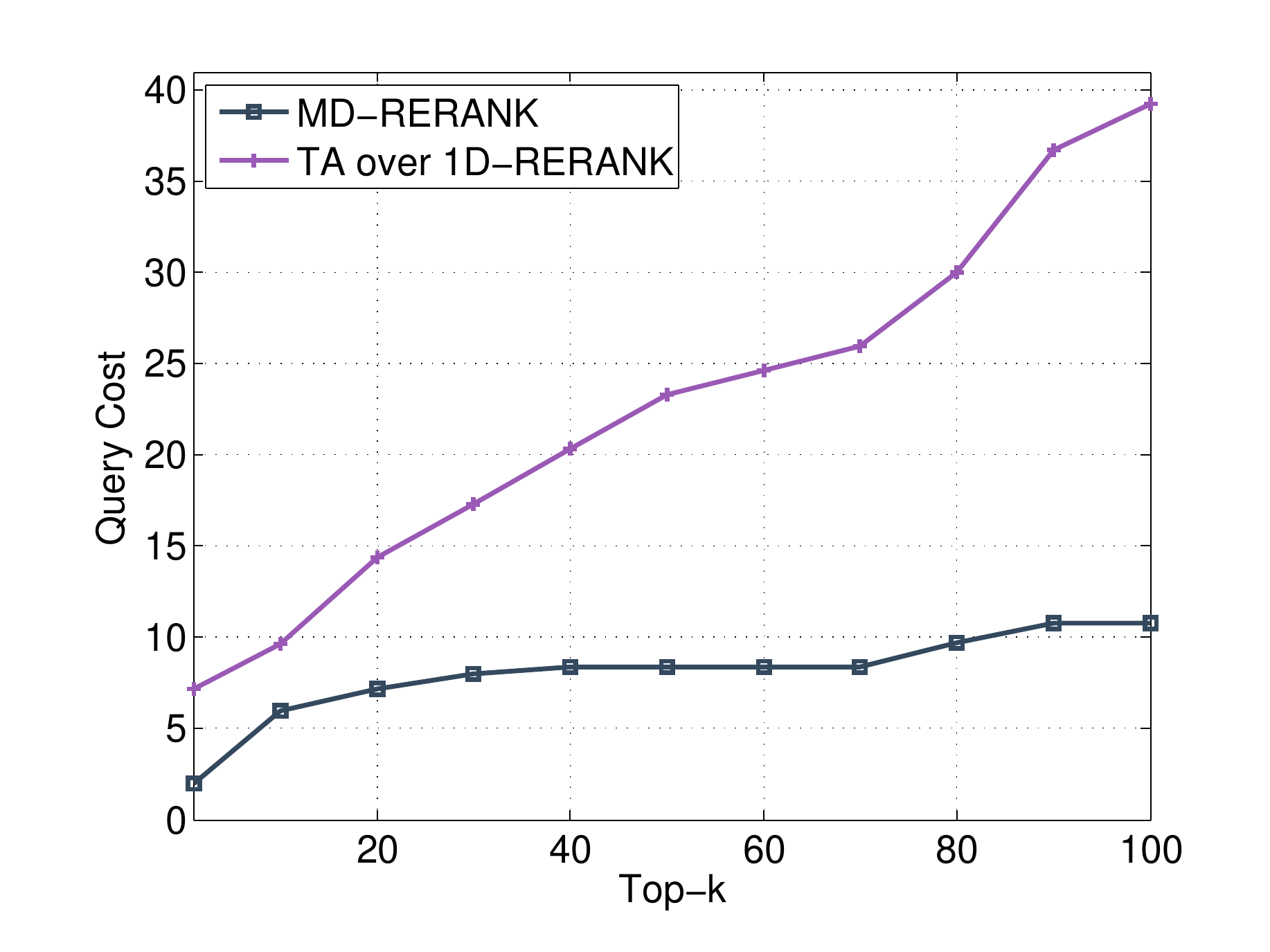}
        \vspace{-2mm}\caption{MD: Top$k$ Query Cost (YA)}
        \label{fig:mdya}
    \end{minipage}
    \hspace{-2mm}
\end{figure*}

\noindent {\bf Impact of Value of $k$:} Figure~\ref{fig:1dk} shows the average (accumulative) query cost for retrieving top-1 to top-10 tuples when the system $k$ varies from 1 to 10.  There are two key observations from the figure: First, our query cost increases (about) linearly with the number of desired top answers, demonstrating its scalability to a large desired answer size. Second, the query cost ,as expected, decreases when the system offers a larger $k$.
 
\noindent {\bf Impact of 1D-RERANK parameters $s$ and $c$:}
Recall from \S~\ref{sec2} that the performance of 1D-RERANK can be parameterized by $s$ and $c$. We conducted two experiments to empirically verify the impact. In the first experiment, we fixed the value of $s$ to $n$ and varied $c$ between $10$ and $n^2$. In the other one, we fixed the value of $c$ to $k\log_2n$ and varied the value of $s$ from $10$ to $n^2$. Figure~\ref{fig:1dsc} shows the average query cost  for both settings. As our theoretical results suggest, setting $c=k\log_2n$ and $s=n$ resulted in the (almost) optimal performance. One can see that further reducing $c$ or increasing $s$ does not have much affect on query cost, yet significantly increases the index size.

\noindent {\bf Impact of Query Order on 1D-RERANK:}
Recall that 1D-RERANK constructs the index on the fly. As such, when queries are issued in different order, the index being maintained may differ.  To test whether the order of user queries have a major effect on the performance of 1D-RERANK, we ran an experiment using SR1 with three query-issuing orders: (1) from low to high selectivity (i.e., from more general to narrower queries), (2) from high to low selectivity, and (3) in a random order. Figure~\ref{fig:1dqo} shows that the query issuance order has a negligible effect on the query cost of 1D-RERANK.

\subsubsection{Online Experiments}
We also conducted two live experiments over Blue Nile and Yahoo! Autos, aiming to retrieve the top-$100$ tuples for each of the user query in the workload. The default system-$k$ for BN and YA are 30 and 15, respectively, with the system ranking function being the default for each website, i.e., descending value of {\it price per carat} for BN and {\it distance from the pre-defined location} for YA.

Figures~\ref{fig:1dbn} and \ref{fig:1dya} show the average query cost for retrieving top-$h$ tuples. 
As expected, 1D-RERANK significantly outperforms the other algorithms for both websites. For BN, while 1D-BINARY performed well in the beginning, it required higher query cost for large values of $h$.
That is because the binary search approach keeps dividing the search area in half until the issued query underflows, thus it is likely to end up with an underflowing query that contains fewer tuples, leading to less saving in the query cost.  For YA, note that 1D-BINARY does not benefit much from the savings and is hence outperformed by 1D-BASELINE.

\subsection{MD Experiments}
\label{sec:expmd}
In this subsection, we compare the performance of MD-RERANK against three baseline methods: the aforementioned ``TA over 1D-RERANK'', as well as MD-BASELINE and MD-BINARY.  Once against, we tested both offline and online settings.
%on DOT dataset and BN, YA websites.
%Since TA algorithm only applies 1D algorithms, it does not need any changes to handle the general positioning assumption.
%The changes suggested in previous subsection suffices. Since the assumption does not affect the discovery of top-$1$ tuple, we only needed to update the algorithms to crawl the tuples with the same value combination as the current top tuple. We used the algorithm from \cite{sheng2012optimal} for crawling.

\noindent
{\bf Constructing Workload of User Preference Queries:} The workload is constructed using a process similar to one described in \S~\ref{sec:exp1d}. However, the ranking functions are constructed by selecting a subset from the set of all ranking attributes and choosing different weights between 0 and 1 for each of them. The workload consists of 32, 12 and 10 queries for DOT, BN and YA, respectively, of which 8, 3 and 2 do not have any filtering conditions.

\subsubsection{Experiments over the Real-world Dataset}

\noindent
{\bf Impact of Database Size and System Ranking Function:} The experimental setup was similar to the 1D experiments in \S~\ref{sec:exp1d}. We evaluated our algorithms for different database sizes and system ranking functions SR1 and SR2. Figures~\ref{fig:mdn1} and \ref{fig:mdn2} shows the results for SR1 and SR2 respectively. In both cases, the algorithm MD-RERANK significantly outperformed all three competing baselines.
One may notice an increase in the query cost of the algorithms when $n$ increases in Figures~\ref{fig:mdn2}, and a decrease in Figures~\ref{fig:mdn1}. That is because when system and user-specified ranking function are anti-correlated, the more tuples database has, the more queries are required to find top tuples for the user-specified ranking function (since more tuples are ranked higher than them based on SR2). The case is vice-verse for SR1.
%\textcolor{red}{Abol: I found the explanation to be confusing and I am not able to rewrite. please ping me}.
%The results show the close performance of MD-BASELINE and MD-BINARY, which is because most of the times the system ranking function returns the tuples in the bottom left corner of the example in Figure~\ref{fig:2D1}.
%The ill-behaving system ranking function affected the performance if MD-BASELINE, such that TA outperformed it. On the other hand, the test queries, beside the strategy of partitioning the region based on the virtual point, resulted in the good performance MD-BINARY.

\noindent {\bf Impact of System-$k$:}
We then varied $k$, the number of tuples returned by the web database and measured the average query cost to retrieve top-$10$ tuples for the query workload. 
Figure~\ref{fig:mdk} shows the results. 
As expected, higher values of system-$k$ required lesser query cost to obtain the top-$10$ tuples.
When $k=1$, our algorithms were not able to use the savings by valid queries resulting in a substantial query cost.

\subsubsection{Online Experiments}
We applied MD-RERANK, as well as TA over 1D-RERANK, to retrieve the top-$100$ tuples for each query in the workload. 
Figure~\ref{fig:mdbn} shows the average query cost for the BN experiment. As shown in the figure, MD-RERANK outperformed TA significantly. 
%MD-RERANK managed to discover all the Top-$100$ tuples with $18$ queries, while the average Top-$10$ discovery query cost for RERANK-TA was $25$.
The results for YA experiment is reflected in Figure~\ref{fig:mdya}. The substantial difference in query cost of the algorithms can be explained by the observation by the negative correlation between the ranking tuples in YA queries (for example the cars with higher mileage are probably cheaper). Hence TA algorithm had to issue many GetNext operations before it finds the top tuples.

\section{Related Work}
\noindent
{\bf Top-$k$ discovery methods} can get divided in three main categories: {\it (sorted/random) access-based} methods, {\it layering-based} approaches, and {\it view-based} techniques.
The first series of algorithms take the advantage of the data access methods. For example, \emph{NRA} \cite{fagin2003} assumes the existence of one sorted list of tuples for each attribute, and finds the Top-$k$ only by exploring the lists, while \emph{TA} \cite{fagin2003} applies both random and sorted access. The more advanced algorithms in this category are \emph{CA} \cite{fagin2003}, \emph{Upper/Pick} \cite{bruno2002}, and \cite{marian2004}.
The next category is the set of algorithms, such as \emph{ONION} \cite{chang2000onion} and \cite{xin}, that pre-process the data and index the layers of extremum tuples that gaurantee including the Top-$k$. %. Constructing the layer for monotonic ranking functions, \emph{ONION} \cite{chang2000onion} uses the set of skyline/skyband tuples, while \cite{xin} considers the linear ranking functions and forms the layers around convex-hull tuples to reduce the index size.
View-based methods such as \emph{PREFER} \cite{PREFER} and \emph{LPTA} \cite{das2006views}, employ the materialized views to increase the efficiency of Top-$k$ discovery process. %Each view can be seen as a sorted list of tuples based on a ranking function. Their goal is to select a subset of views that can help minimizing the cost.
While prior work focused on minimizing the storage overhead of indices/materialized views and the computational overhead of processing top-$k$ queries, we have to focus on minimizing the number of queries issued to the underlying database.  This fundamentally different data access model also leads to a different cost model. For example, many prior work, such as \cite{fagin2003} and \cite{CH02}, assume a separate cost for accessing each attribute and/or evaluating each predicate in the top-$k$ query, %: e.g., the classic TA algorithm \cite{fagin2003} considers a separate random-access cost for attributes, while \cite{CH02} aims to minimize the number of ``probes'' for evaluating expensive predicates in user queries. 
while in our problem {\em all} attributes of a tuple are returned at once. % from the underlying database, requiring no extra cost for accessing each additional attribute.

\noindent
{\bf Hidden Databases}
Most of the prior works on the hidden databases relate to sampling, crawling the database, and aggregate estimation. 
%Sampling a hidden database enables answering aggregate estimates on the database.
Prior works such as \cite{dasgupta2007random, wang2011effective} propose efficient algorithms for collecting unbiased low-variance random samples of a given hidden databasee and \cite{dasgupta2010unbiased, liu2014aggregate} provide unbiased aggregate estimators.
%\cite{dasgupta2007random, dasgupta2009leveraging, dasgupta2010turbo} propose efficient algorithms for collecting unbiased random samples of a given hidden databasee. \cite{liu2012stratified, wang2011effective} improve prior works by reducing the variance of the estimate.
%\cite{dasgupta2010unbiased} provides an unbiased estimator for COUNT and SUM aggregates for static databases, while \cite{liu2014aggregate} extends them for dynamic hidden databases. 
While \cite{sheng2012optimal, raghavan2000crawling, madhavan2008google} aim toward crawling the whole hidden database, \cite{skylinediscovery} only crawls the maxima index.

\noindent
{\bf Top-$k$ queries over Hidden Databases}
% breaking top-k barrier
As the best of our knowledge, this is the first paper on reranking the query results of a hidden database. The only prior work about Top-$k$ in hidden databases is \cite{thirumuruganathan2013breaking}. Assuming the full knowledge of the system ranking function and attribute domains, its goal is to go beyond the Top-$k$ limitation of the database interface, by partitioning the query space.% in two queries.
\label{sec-related}
\section{final remarks} \label{sec-final}
In this paper, we introduced a novel problem of query reranking, a third-party service that takes a client-server database with a proprietary ranking function and enables query processing according to any user-specified ranking function.  To enable query reranking while minimizing the number of queries issued to the underlying database, we develop 1D-RERANK and MD-RERANK for user-specified ranking functions that involve only one attribute and any arbitrary set of attributes, respectively. Theoretic analysis and extensive experimental results on real-world databases, in offline and online settings, demonstrate the effectiveness of our techniques and their superiority over baseline solutions.

\bibliographystyle{abbrv}
\vspace{-1mm}
\bibliography{ref}

\end{document}